\documentclass[11pt]{article}

\usepackage{amsmath,amsthm,nicefrac}
\usepackage{amsfonts, amstext}

\usepackage{comment}

\usepackage{typearea}
\typearea{14}
\usepackage[font={small,it}]{caption}

\usepackage{setspace}
\usepackage[compact]{titlesec}

\usepackage{enumitem}
\usepackage[breaklinks]{hyperref}
\hypersetup{colorlinks=true,%
    citebordercolor={.6 .6 .6},linkbordercolor={.6 .6 .6},%
    citecolor=blue,urlcolor=black,linkcolor=blue,pagecolor=black}

\usepackage[nameinlink,capitalise]{cleveref}
\Crefname{algocf}{Algorithm}{Algorithms}
\crefname{algocfline}{line}{lines}
\Crefname{invariant}{Invariant}{Invariants}
\Crefname{claim}{Claim}{Claims}
\Crefname{subclaim}{Subclaim}{Subclaims}

\usepackage{epsfig}
\usepackage{amsthm,amssymb}
\usepackage{amsthm,amssymb}

\usepackage{mathrsfs}
\usepackage{xspace}
\usepackage{soul}
\usepackage{latexsym}
\usepackage{bm}

\usepackage{framed}

\usepackage[dvipsnames]{xcolor}
\definecolor{DarkGray}{rgb}{0.66, 0.66, 0.66}
\definecolor{DarkPowderBlue}{rgb}{0.0, 0.2, 0.6}
\definecolor{fluorescentyellow}{rgb}{0.8, 1.0, 0.0}

\usepackage[ruled,vlined,linesnumbered,algonl]{algorithm2e}
\SetEndCharOfAlgoLine{}
\SetKwComment{Comment}{\footnotesize$\triangleright$\ }{}

\SetCommentSty{mycommfont}

\usepackage{thmtools,thm-restate}

\usepackage{algorithmic}
\usepackage{mathtools}

\allowdisplaybreaks

\newcounter{note}[section]

\sethlcolor{fluorescentyellow}

\newcommand{\initOneLiners}{%
    \setlength{\itemsep}{0pt}
    \setlength{\parsep }{0pt}
    \setlength{\topsep }{0pt}
}

\newtheorem{theorem}{Theorem}[section]
\newtheorem{lemma}[theorem]{Lemma}

\theoremstyle{definition}
\newtheorem{definition}[theorem]{Definition}

\theoremstyle{remark}

\makeatletter
\renewcommand{\theinvariant}{(I\@arabic\c@invariant)}
\makeatother

\newcommand{\floor}[1]{\left\lfloor #1 \right\rfloor}
\newcommand{\ceil}[1]{\left\lceil #1 \right\rceil}

\newcommand{\mailto}[1]{\href{mailto:#1}{\texttt{#1}}}

\newcommand{\NN}{\mathbb N}

\newcommand{\abs}[1]{\left\lvert #1 \right\rvert}

{
\end{proof}}

\newcommand{\SC}{\mathcal C}

\newcommand{\tst}{t^\star}

\newcommand{\eat}[1]{}

\usepackage{tikz}
\usetikzlibrary{arrows.meta,positioning,calc}

\begin{document}

    \title{Hallucination Rates in Language Generation}

    \author{
        Debmalya Panigrahi\thanks{Department of Computer Science, Duke University, 308 Research Drive, Durham, NC 27708, USA. Research supported in part by NSF grant CCF-2329230. Email: {\tt debmalya@cs.duke.edu}.} \and
        Fan Wei\thanks{Department of Mathematics, Duke University, 120 Science Drive, Durham, NC 27710, USA. Research supported by NSF grants DMS-2401414 and DMS-2601663. Email: {\tt fan.wei@duke.edu}.}  \and
        Ian Zhang\thanks{Duke University. Email: \mailto{ian.zhang@duke.edu}.}
    }

    \date{July 2026}

    \maketitle

    \begin{abstract}
        Language generation in the limit is an elegant model introduced by Kleinberg and Mullainathan \cite{KM24} to
        formally study language generation by an algorithm that learns solely based on example strings.
        In this model, an algorithm is said to correctly generate from a language
        if it never makes an error after some finite time.
        In contrast, even sophisticated language models are known to regularly hallucinate in practice.
        In this paper, we initiate the study of language generation in the limit with (infinite) hallucination,
        i.e., the algorithm may generate incorrect strings infinitely often,
        but the errors occur at a limited rate (possibly even with $0$-measure).

        We first show that hallucination, even at rate $0$, makes generation in the limit strictly more powerful,
        i.e., there are language collections that cannot be generated with finite error
        but can be generated with infinite error, even when the errors occur at a $0$-measure set of time-steps.
        Furthermore, we show a surprising conceptual separation between algorithms that are allowed to make finite
        and infinite errors: while all countable language collections are known to be generatable with finite error,
        we show that there exists a strict hierarchy of (uncountable) language collections
        when characterized by the hallucination rate.
        Indeed, this distinction extends to the breadth of generation,
        which characterizes the fraction of the target language that is generated by the algorithm.
        While all countable collections can be generated with the optimal breadth $1/2$~\cite{KW26},
        we show a strict separation between generatable language collections
        at {\em every} breadth level and allowed rate of hallucination.
        Finally, we study generation in the limit {\em without repetition},
        where the algorithm is not permitted to repeat any strings.
        This allows us to compare the set of correct strings generated by the algorithm to the set of incorrect strings,
        rather than the fraction of time-steps when it generates correctly or incorrectly.
        Once again, we demonstrate the existence of a strict hierarchy that separates language collections
        that can be generated without repetition at every rate of hallucination and every breadth of generation.
        Taken together, these results reveal rich structure in language collections
        that can be generated in the limit with hallucination, and establish the rate of hallucination
        as an important parameter in the theoretical study of language generation.
    \end{abstract}

    \pagenumbering{gobble}
    \clearpage

    \pagenumbering{arabic}

    \section{Introduction}\label{sec:intro}The study of language generation has emerged as an important goal in both theory and practice with the advent and popularity of large language models (LLMs). A prominent direction of theoretical research in this domain is that of {\em language generation in the limit}, where the goal is to study the feasibility of generating (formal) languages in an unsupervised setting solely based on example strings provided to the generation algorithm. This framework draws its inspiration from classical work of Gold in the 1960s~\cite{Gol67}, and was formalized recently in the language generation setting by the influential work of Kleinberg and Mullainathan (KM)~\cite{KM24}. They showed that for any countable collection of languages $\SC$, given an example string from a target language $K\in \SC$ in every time-step, an algorithm can generate correct unseen strings from $K$ after a finite number of time-steps. This was surprising because it went against intuition derived from strong impossibility results shown by Angluin in the 1980s~\cite{Ang80,Ang80a} for the same model, but for the more challenging task of {\em identifying} the target language $K$ rather than simply {\em generating} correctly from it. Spurred by the positive result in the KM paper, the language generation in the limit problem has gained rapid popularity in the last couple of years, with various works branching out in several important directions such as understanding the ``breadth'' of the generated set in the target language~\cite{KW25,KMV25,KW26}, robustness of the algorithm to ``noise'' or ``omissions'' in the examples~\cite{RR25,LRT25,MVYZ26,BPZ26,LZ26,KW26}, finer-grained categorization of generatable language collections based on different parameters~\cite{CP25,KMV25,KMV26}, and various other directions~\cite{CP26,ABCK26,AAK26,HP26,LRT26}.

A key observation in language models in practice is the phenomenon of {\em hallucination}, where the language model occasionally generates incorrect strings interleaved with correct ones. Language models, even the most sophisticated ones, routinely produce hallucinated information, which has inspired rigorous studies in both theory~\cite{KV24,KNVZ25} and practice~\cite{HYMZFWCPF25,FKKG24} to understand this phenomenon and study its implications. In this paper, we initiate the study of hallucination in language generation in the limit. The original KM model stipulates that the number of incorrect strings must be finite for a language to be generated correctly by an algorithm. We relax this requirement and allow infinite hallucination, only constraining the rate at which the algorithm hallucinates, and call this {\em language generation with hallucination}.
In particular, we say that a collection of languages can be $\beta$-generated in the limit for some $\beta \in [0, 1]$
if there exists an algorithm whose rate of hallucination in the limit is at most $\beta$, i.e., in the limit, the algorithm generates incorrectly during at most a $\beta$-fraction of the time-steps. Note that this sharply distinguishes generation with hallucination from the KM model: 
an algorithm can make an infinite number of errors even for $\beta = 0$, but one in the KM model can only make a finite number of errors.

The first question is the power of hallucination---are there language collections that cannot be generated in
the KM model but become generatable when we allow the algorithm to hallucinate? Clearly, such language collections,
if they exist, must be uncountable since all countable collections can be generated in the KM model~\cite{KM24}.
We answer this question in a strong positive sense: we show that there are uncountable language collections that cannot
be generated in the KM model, but become generatable even when the set of hallucinated time-steps has zero measure in the time horizon.
\begin{restatable}{theorem}{zero}
    \label{thm:zero}
    There exists a collection $\SC_0$ that is $0$-generatable in the limit,
    but is not generatable with finite error.
\end{restatable}

This raises an interesting question: does the {\em rate of hallucination},
namely the limit on the ratio of time-steps when the algorithm hallucinates to the total number of time-steps,
determine the generatability of language collections?
Indeed, we give a strong positive answer to this question as well: we show that there is a strict hierarchy on language collections based on the rate of hallucination.
\begin{restatable}{theorem}{real}
    \label{thm:real}
    For any $\beta \in (0, 1]$, there exists a collection $\SC_{\beta}$ that is $\beta$-generatable
    in the limit, but is not $\beta'$-generatable in the limit for any $\beta' < \beta$.
\end{restatable}

Having established the power of hallucination, one may wonder if all language collections can be generated with a sufficiently high rate of hallucination. We show that this is not true: 
in fact, we give a language collection where no algorithm can generate even a {\em single} correct string.
\begin{restatable}{theorem}{ungenall}
    \label{thm:ungen}
    There exists a collection $\SC$ that is not $\beta$-generatable in the limit for any $\beta < 1$.
    Furthermore, for any algorithm $G$, there exists an adversary strategy such that $G$
    never generates a single correct string from some target language $K\in \SC$.
\end{restatable}

The above results create a precise hierarchy of language collections starting with countable collections that can be generated with finite error (by the KM result~\cite{KM24}) all the way up to language collections that cannot be generated at all, even with an arbitrarily large rate of hallucination.

Next, we study the tradeoff between the rate of hallucination and the breadth of language generation. One shortcoming of the algorithm in~\cite{KM24} was the existence of countable language collections that were correctly generated by the algorithm, but the strings output by the algorithm formed a subset of zero measure in the actual target language. This motivated Kleinberg and Wei (KW)~\cite{KW25} to introduce the notion of breadth (or density) of generation, which parametrizes the fraction of the target language that is actually generated by the algorithm.
They define two related notions of breadth, {\em upper} and {\em lower} density, which represent upper and lower limits on the fraction of the target language generated by the algorithm.
In a follow-up paper~\cite{KW26}, they also gave a new algorithm which
generates all countable language collections with a lower density of $1/2$,
which is optimal since the algorithm is not allowed to repeat strings provided as examples,
and these strings can be completely disjoint from the algorithm's strings in the worst case.

The KW result~\cite{KW26} implies the absence of a strict hierarchy of language collections by breadth of generation:
for any $\alpha, \alpha'$ satisfying $0 < \alpha < \alpha' \le 1/2$,
there is no countable language collection that can be generated with breadth $\alpha$ but not with breadth $\alpha'$.
Surprisingly, we show that the situation is very different for language generated with hallucination.
We first show that there exists a language collection that cannot be generated with upper density better than $0$
regardless of the amount of hallucination allowed, although it can be generated in the KM model (with $0$ upper density).
\begin{restatable}{theorem}{densityzero}
    \label{thm:density_zero}
    There exists a collection $\SC$ that is generatable in the limit with finite error and upper density $0$, but it cannot be $\beta$-generated with upper density strictly greater than $0$ even with arbitrary $\beta \le 1$.
\end{restatable}

We further explore the interplay of breadth and rate of hallucination, and show the existence of a strict hierarchy at every rate of hallucination and breadth:
for any $\beta > 0$ and any $0 \le \alpha \le 1/2$,
we give a language collection that can be generated with lower density $\alpha$ when hallucinating at rate $\beta$,
but can neither be generated with upper density greater than $\alpha$ (regardless of the rate of hallucination),
or with a rate of hallucination less than $\beta$ (regardless of the breadth of generation).
(Note that showing a lower bound against generation with upper density is stronger than a lower bound against generation with lower density because generation with lower density $\alpha$ implies generation with upper density $\alpha$).

\begin{restatable}{theorem}{densityhierarchy}
    \label{thm:two_sided_hierarchy}
    For any density of generation $\alpha \in [0, 1/2]$ and for any $\beta \in (0, 1]$,
    there exists a collection $\SC_{\beta, \alpha}$ that is
    $\beta$-generatable in the limit with lower density $\alpha$,
    but is not $\beta'$-generatable in the limit for any $\beta' < \beta$ (even with upper density $0$),
    and is not $\gamma$-generatable in the limit with upper density $\alpha' > \alpha$ for any $\gamma \le 1$.
\end{restatable}

Since the rate of hallucination encapsulates the fraction of time-steps when the algorithm generates incorrect strings, a natural question is whether a small rate of hallucination implies that the algorithm outputs a small fraction of incorrect strings.
In general, this is not true. This is because the original KM model~\cite{KM24} (and our previously stated results) allows the generation algorithm to output the same string multiple times, as long as that string has not been provided as an example. This means that an algorithm might have a small rate of hallucination while also having a large fraction of incorrect strings in its output set, by repeating correct strings more frequently than incorrect ones. This motivates us to study the language generation problem in a modification of the KM model that explicitly prohibits the algorithm from repeating the same string multiple times, even if it has not been provided as an example. We call this {\em language generation without repetitions}. Interestingly, this modification does not affect the central KM result: their generation algorithm for any countable collection of languages does not need to repeat any string. But, once again, we observe a more fine-grained picture when we try to categorize uncountable language collections in terms of generation with hallucination. We demonstrate a strict hierarchy of language collections with the allowed rate of hallucination when the algorithm is not allowed to repeat output strings.
\begin{restatable}{theorem}{boxreal}
    \label{thm:box}
    For any $\beta \in (0, 1]$, there exists a collection $\SC_{\beta}$ that is $\beta$-generatable in the limit without
    repetitions, but is not $\beta'$-generatable in the limit without repetitions for any $\beta' < \beta$.
\end{restatable}

Finally, we consider the interplay between the breadth of generation and the hallucination rate in language generation without repetition.
Analogous to \cref{thm:two_sided_hierarchy}, we show the existence of a strict hierarchy for generation without repetition at every breadth and hallucination rate. (Note, however, that the language collections used to show \Cref{thm:two_sided_hierarchy} and \Cref{thm:density_hierarchy_without} are different since they correspond to different models of language generation, with and without repetition respectively.)

\begin{restatable}{theorem}{densityhierarchywithout}
    \label{thm:density_hierarchy_without}
    For any density of generation $\alpha \in [0, 1/2]$ and for any $\beta \in (0, 1]$,
    there exists a collection $\SC_{\beta, \alpha}$ that is
    $\beta$-generatable in the limit without repetitions with lower density $\alpha$,
    but is not $\beta'$-generatable in the limit without repetitions for any $\beta' < \beta$
    (even with upper density $0$), and is not $\gamma$-generatable in the limit without repetitions
    with upper density $\alpha' > \alpha$ for any $\gamma \le 1$.
\end{restatable}

\subsection{Related Work}\label{sec:related}

There have been many recent works studying the KM language generation model.
As discussed previously, one line of research has focused on formalizing the tradeoff
between correctness and breadth of generation.
Various notions of breadth have been formalized, which study the ``fraction'' of strings
in the target language that the algorithm eventually
generates~\cite{CP25,KMV25,KW25,KW26,MVYZ26,KMV26,KW26a,GMDT26,LHSGJG26}.
Another line of work has studied the impact of noise or omissions,
where the adversary is allowed to either add or remove strings
from its enumeration of the target language~\cite{RR25,MVYZ26,BPZ26,LZ26,KW26,HP26}.
Other work has investigated whether generatability is closed under finite unions~\cite{HKMV25,BPZ26}.

Another line of research has investigated notions of uniform and non-uniform generation,
which enforce additional constraints on the (finite) time at which the
algorithm must generate correctly, or the number of mistakes the algorithm can
make~\cite{LRT25,CP25,CP26,ABCK26,KPR26,JKO26,BPZ26}.
There has also been a wide range of work studying natural variants of language generation,
or settings that impose additional constraints on the adversary or
algorithm~\cite{PRR25,AAK26,LRT26,RVS26,MVYZ26a,LHJG26,PF25,LHJG26a,FPPSS26}.
Finally there has been work that revisits the problem of language
identification~\cite{PSV26,CPT26,CKP26,CKP27}.

\paragraph{Roadmap.}
We introduce basic definitions used throughout the paper in \Cref{sec:prelim}.
Then, we give our results for generation with hallucination allowing repetitions
(\Cref{thm:zero,thm:real,thm:ungen,thm:density_zero,thm:two_sided_hierarchy}) in \Cref{sec:with_repeats}
and those for generation without repetitions (\Cref{thm:box,thm:density_hierarchy_without}) in
\Cref{sec:without_repeats}. Finally, we conclude with some closing remarks in \Cref{sec:closing}.

    \section{Preliminaries}\label{sec:prelim}This section introduces the standard definitions for language generation in the limit,
along with the new definitions for generation with hallucination.
We use $U$ to denote the countably infinite universe containing all strings.
Every language $L$ is an infinite subset of the universe $U$.
A collection of languages $\SC$ is a subset of the power set of $U$, i.e., $\SC \subseteq 2^U$.
In general, a collection of languages $\SC$ can be countable or uncountable.
It will sometimes be convenient to use the set of natural numbers (positive integers)
to define the universe (this is without loss of generality up to renaming
since there is a bijective map between any two countably infinite sets).
We denote the set of natural numbers by $\NN = \{1, 2, \ldots\}$
and the first $n$ natural numbers $\{1, 2, \ldots, n\}$ by $[n]$.

\subsection{Generation in the Limit}

In the general setup~\cite{KM24}, the universe $U$ and the collection of languages $\SC$
are fixed and known to the algorithm.
In any run of the algorithm, the adversary selects a target language $K \in \SC$;
this choice is unknown to the algorithm.
The adversary presents the strings of $K$ in an enumeration $x_1$, $x_2, x_3, \ldots$,
where each string $x_t$ is contained in $K$, and for every string $y \in K$,
there exists some time $t$ for which $y = x_t$.
In other words, the adversary needs to eventually enumerate all the strings in $K$.
We note that in previous literature, the adversary has sometimes been allowed to repeat strings in its enumeration.
However, since the rate of hallucination is defined in terms of time-steps,
allowing the adversary to repeat strings can artificially boost the rate of hallucination
because the adversary can postpone providing new example strings to the algorithm.
For this reason, we enforce in this paper that the adversary must provide a distinct example string in each time-step.
We also note that it was shown by~\cite{BPZ26} that generation in the limit is equivalent regardless of whether
the adversary is allowed to repeat strings in its enumeration.

At each time-step $t$, the algorithm takes as input the set of strings enumerated by the adversary so far
and outputs a string $z_t$ that is distinct from $x_1, x_2, \ldots, x_t$.
The goal is that after some finite time $\tst$, all strings
$z_t$ for $t \ge \tst$ have to be correct \emph{unseen} strings from the target language $K$.
In other words, $z_t$ must be in $K$ for all $t \ge \tst$.
Note that unlike the adversary, the algorithm is allowed to output the same string multiple times,
as long as the adversary has not output that string in its enumeration yet.
Later, we will also consider the setting where the algorithm cannot repeat strings, similar to the adversary.

\begin{definition}[Generator algorithm]
    A generator algorithm is a function $U^* \to U$ which takes as input a finite ordered
    set of strings $x_1$, \dots, $x_t$, and outputs a string $z_t$ such that $z_t \neq x_{t'}$
    for any $t' \le t$.
\end{definition}

\begin{definition}[Generation in the limit~\cite{KM24}]
    An algorithm $G$ generates in the limit for a collection $\SC$
    if for any $K \in \SC$ and any enumeration $x$ of $K$,
    there exists a time $\tst$ such that for all $t \ge \tst$,
    the generated string $z_t$ at time $t$ is in $K$.
\end{definition}

Given an adversary enumeration $x_1$, $x_2$, \dots\ and corresponding algorithm outputs
$z_1$, $z_2$, \dots, we say that a string $y \in U$ is \emph{unrevealed} at time $t$
if $y$ has not been output by either the adversary or algorithm before time $t$,
i.e., $y \notin \{x_1, z_1, \dots, x_{t-1}, z_{t-1}\}$.

\subsection{Generation with Breadth}
We now introduce the notion of breadth of generation as defined in~\cite{KW25}.
To do so, we first impose an arbitrary fixed ordering on the strings of $U$.
Let $L$ and $L'$ be two languages of $U$.
Let the strings of $L'$ be listed in order as $L' = \{\ell'_1, \ell'_2, \ell'_3, \dots\}$.
The (asymptotic) upper density of $L$ in $L'$ is
\[\mu_{\text{up}}(L, L') = \limsup_{n \to \infty} \frac{\abs{L \cap \{\ell'_1, \ell'_2, \dots, \ell'_n\}}}{n},\]
and the (asymptotic) lower density of $L$ in $L'$ is
\[\mu_{\text{low}}(L, L') = \liminf_{n \to \infty} \frac{\abs{L \cap \{\ell'_1, \ell'_2, \dots, \ell'_n\}}}{n}.\]
We define the breadth of generation in terms of the density of the algorithm's outputs
within the target language $K$.
For any adversary enumeration $E = \{x_1, x_2, x_3, \dots\}$ and algorithm $G$,
let $O(E, G) = \{z_1, z_2, z_3, \dots\}$ be the set of strings that the algorithm ever outputs
when given the enumeration $E$.

\begin{definition}
    For any $\alpha \in [0, 1]$ and collection $\SC$, an algorithm $G$ generates with upper density $\alpha$
    if for every target language $K \in \SC$ and adversary enumeration $E$ of $K$, we have
    \[\mu_{\text{up}}(O(E, G), K) \ge \alpha.\]
    Similarly, an algorithm $G$ generates with lower density $\alpha$ if we have for every $K$ and $E$ that
    \[\mu_{\text{low}}(O(E, G), K) \ge \alpha.\]
\end{definition}

\subsection{Generation with Hallucination}
We now introduce our new definitions of language generation with hallucination.
Consider any target language $K$, an adversary enumeration $x_1, x_2, \ldots$,
and an algorithm that outputs $z_1, z_2, \ldots$. 
Let $V = \{t \mid z_t \in K \setminus \{x_1, \dots,  x_t\}\}$
be the set of time-steps at which the algorithm's output is a correct unseen
string from the target language,
and $W = \NN \setminus V$ be the set of time-steps at which the algorithm's output is incorrect.
Also, denote the prefixes of these sets for the first $t$ time-steps as
$V_t = V \cap [t]$ and $W_t = W \cap [t]$ respectively.

We define $\beta$-generation in the limit to mean that in the limit, an algorithm
generates incorrectly during at most $\beta$ fraction of time-steps.
\begin{definition}
    For any $\beta \in [0, 1]$ and collection $\SC$, an algorithm $\beta$-generates in the limit
    if for every target language $K \in \SC$ and adversary enumeration $E$ of $K$, we have
    \[\limsup_{t \to \infty} \frac{\abs{W_t}}{t} \le \beta.\]
\end{definition}
Note that requiring $\limsup_{t \to \infty} \frac{\abs{W_t}}{t} \le \beta$ is a more stringent
condition than requiring $\liminf_{t \to \infty} \frac{\abs{W_t}}{t} \le \beta$.
We adopt the stronger requirement because it corresponds more naturally to enforcing that
the maximum rate of error is $\beta$.
An even stronger condition would be to require $\sup_{t \to \infty} \frac{\abs{W_t}}{t} \le \beta$,
but using $\limsup$ is more natural in the KM model, which requires the algorithm to generate correctly in the limit.

It is straightforward to see that we can also define $\beta$-generation in the limit
in terms of the fraction of correct time-steps.
\begin{lemma}
    An algorithm $\beta$-generates in the limit if and only if
    \[\liminf_{t \to \infty} \frac{\abs{V_t}}{t} \ge 1 - \beta.\]
\end{lemma}
\begin{proof}
    Since $\liminf_{t \to \infty} (-y_t) = - \limsup_{t \to \infty} y_t$ for any real sequence $\{y_t\}$,
    we have
    \[\limsup_{t \to \infty} \frac{\abs{W_t}}{t} \le \beta
        \iff \liminf_{t \to \infty} \left( 1 - \frac{\abs{W_t}}{t} \right) \ge 1 - \beta
        \iff \liminf_{t \to \infty} \frac{\abs{V_t}}{t} \ge 1 - \beta.\qedhere\]
\end{proof}

We say that a collection is $\beta$-generatable in the limit with breadth $\alpha$
if there exists a single algorithm $G$ that both $\beta$-generates in the limit and generates with breadth $\alpha$.

\subsection{Generation without Repetitions}
In the previous definitions,
the generation algorithm is allowed to repeat strings that have not yet been presented by the adversary.
This creates the possibility that an algorithm with a small rate of hallucination
can actually have a large fraction of incorrect strings in its output set
by repeating correct strings more frequently than incorrect ones. 
To directly bound the fraction of incorrect strings produced by an
algorithm, we now consider generation without repetitions, where we require that
the algorithm outputs a distinct string in every time-step.
Under this definition, the fraction of time-steps at which the algorithm generates correct strings
exactly matches the fraction of correct strings in the algorithm's output set.

\begin{definition}[Generator algorithm without repetitions]
    A generator algorithm without repetitions is a function $U^* \to U$ which takes as input a finite ordered
    set of strings $x_1, \ldots, x_t$, and outputs a string $z_t$ such that $z_t \neq z_{t'}$
    for any $t' < t$ and $z_t \neq x_{t'}$ for any $t' \le t$.
\end{definition}

The definitions of breadth of generation and hallucination rate given above seamlessly extend to the case of generation without repetition. However, as we mentioned above, these parameters can now be interpreted in terms of the {\em set of strings} output by the algorithm which forms a bijection with the time-steps, since the algorithm is not allowed to repeat any string.

    \section{Generation with Hallucination Allowing Repetitions}\label{sec:with_repeats}In this section, we study the role of (infinite) hallucinations in the
language generation in the limit model (with repetitions allowed).
Recall that we say that an algorithm $\beta$-generates a language collection
if it generates in the limit with rate of hallucination at most $\beta$.

\subsection{Generation with Hallucination Rate $0$}

Our first result shows that there are language collections that are $0$-generatable in the limit,
but are not generatable in the limit in the standard KM model.
This shows that allowing an infinite but measure zero number of hallucinations
is more powerful than the standard model which only allows a finite number of errors.

\zero*

We first describe the collection $\SC_0$.
The universe of the collection is described by an infinite rooted tree $T$ where every node has exactly $2$ children.
We label the nodes of $T$ with strings of $\bigcup_{i \in \NN} [2]^i$,
where the root of $T$ is labeled by $\emptyset$, and we inductively assign
$(u_1, \dots, u_{i})$ to label the $u_{i}$-th child of $(u_1, \dots, u_{i-1})$.
We can think of the label of a node $v$ as encoding the unique path from the root to $v$.

Let the depth of a node $v$ be denoted by $d(v)$ where the depth of the root is $1$,
and the depth of any other node is one more than the depth of its parent.
For each node $v$, we will associate with $v$ a set $S_v$ of size $d(v)$ which contains strings
that are disjoint from any other set $S_u$.
Formally, we have $S_v = \{vj \mid j \in [d(v)]\}$, where $vj$
is the concatenation of $v$ with $j$.
\Cref{fig:zero-tree} depicts the first $3$ layers of the tree.

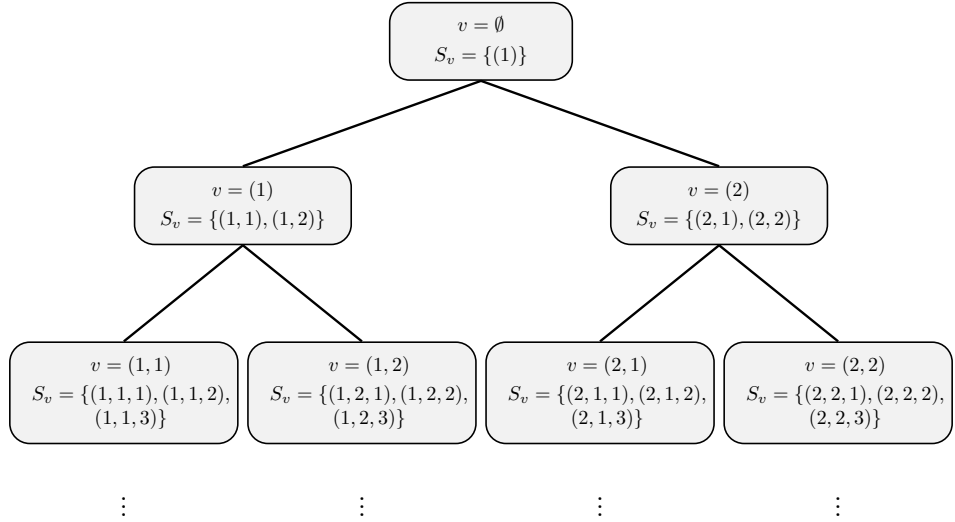
\begin{figure}[t]
    \centering
    \scalebox{.75}{
        \begin{tikzpicture}[
            edge/.style={draw, very thick},
            info/.style={
                draw,
                rounded corners=10pt,
                thick,
                fill=gray!10,
                align=center,
                inner sep=6pt,
                font=\small
            },
            dots/.style={font=\Large}
        ]

            \node[info, text width=2.8cm] (root) at (0,0)
                {$\begin{array}{c}
                      v=\emptyset\\[1mm]
                      S_v=\{(1)\}
            \end{array}$};

            \node[info, text width=3.4cm] (left) at (-4.2,-2.9)
                {$\begin{array}{c}
                      v=(1)\\[1mm]
                      S_v=\{(1,1),(1,2)\}
            \end{array}$};

            \node[info, text width=3.4cm] (right) at (4.2,-2.9)
                {$\begin{array}{c}
                      v=(2)\\[1mm]
                      S_v=\{(2,1),(2,2)\}
            \end{array}$};

            \node[info, text width=3.6cm] (ll) at (-6.3,-6.2)
                {$\begin{array}{c}
                      v=(1,1)\\[1mm]
                      S_v=\{(1,1,1),(1,1,2),\\
                      (1,1,3)\}
            \end{array}$};

            \node[info, text width=3.6cm] (lr) at (-2.1,-6.2)
                {$\begin{array}{c}
                      v=(1,2)\\[1mm]
                      S_v=\{(1,2,1),(1,2,2),\\
                      (1,2,3)\}
            \end{array}$};

            \node[info, text width=3.6cm] (rl) at (2.1,-6.2)
                {$\begin{array}{c}
                      v=(2,1)\\[1mm]
                      S_v=\{(2,1,1),(2,1,2),\\
                      (2,1,3)\}
            \end{array}$};

            \node[info, text width=3.6cm] (rr) at (6.3,-6.2)
                {$\begin{array}{c}
                      v=(2,2)\\[1mm]
                      S_v=\{(2,2,1),(2,2,2),\\
                      (2,2,3)\}
            \end{array}$};

            \draw[edge] (root.south) -- (left.north);
            \draw[edge] (root.south) -- (right.north);
            \draw[edge] (left.south) -- (ll.north);
            \draw[edge] (left.south) -- (lr.north);
            \draw[edge] (right.south) -- (rl.north);
            \draw[edge] (right.south) -- (rr.north);

            \node[dots] at (-6.3,-8.1) {$\vdots$};
            \node[dots] at (-2.1,-8.1) {$\vdots$};
            \node[dots] at (2.1,-8.1) {$\vdots$};
            \node[dots] at (6.3,-8.1) {$\vdots$};
        \end{tikzpicture}
    }
    \caption{The first three layers of the rooted binary tree used in the proof of \cref{thm:zero}.
    Each node $v$ is labeled together with its associated set $S_v = \{vj : j \in [d(v)]\}$.}
    \label{fig:zero-tree}
\end{figure}

The universe $U = \bigcup_{v \in T} S_v$ is simply the union of the sets at each node in the tree.
$U$ is countable because there are a countable number of nodes in $T$,
and the set $S_v$ is finite for each $v$.
Languages $L \in \SC_0$ will be constructed from infinite rays in the tree $T$
that originate at the root.
Let $P = (v_1, v_2, v_3, \dots)$ be any infinite ray in $T$ where $v_1 = \emptyset$
is the root of $T$, and $v_j$ is a child of $v_{j-1}$ for any $j \ge 2$.
The collection $L_P = \bigcup_{v \in P}S_{v}$ is the union of the sets along each node
in the ray $P$.
Let $I = \{(v_1, v_2, v_3, \dots) \mid v_1 = \emptyset, \text{ and $v_j$ is a child of $v_{j-1}$ for all $j \ge 2$}\}$
be the set of all infinite rays in $T$ that originate at the root.
Then the collection of languages is \[\SC_0 = \{L_P \mid P \in I\}.\]
We now give an algorithm that $0$-generates in the limit for $\SC_0$.

\begin{lemma}
    \label{lem:zero_algo}
    There exists an algorithm that $0$-generates in the limit for $\SC_0$.
\end{lemma}

Intuitively, if the adversary outputs a single string from some set $S_v$,
then the entirety of $S_v$ must be contained in the target language,
allowing the algorithm to safely output any string in $S_v$.
Thus, the only times when the algorithm cannot correctly generate a string is when
for every node $v$, either the adversary has output no strings from $S_v$
or has output all of the strings from $S_v$.
Since the sizes of $S_v$ are increasing, we show that the algorithm
only generates incorrectly at a density $0$ fraction of times.

For any string $x \in U$, we define $n(x)$ to be the unique node $v \in T$
such that $x \in S_v$.

\begin{lemma}
    \label{lem:zero_contain}
    For any language $L \in \SC_{0}$ and string $x$, if $x \in L$, then $S_{n(x)} \subseteq L$.
\end{lemma}
\begin{proof}
    Fix an arbitrary ray $P \in I$ such that $x \in L_P$.
    By definition, $L_P = \bigcup_{v \in P} S_{v}$.
    Thus if $x \in L$, that implies $n(x) \in P$, so $S_{n(x)} \subseteq L$.
\end{proof}

We now define the algorithm that $0$-generates in the limit for $\SC_0$.
At some time $t$, if there exists a string in $S_{n(x_t)} \setminus \{x_1, \dots, x_t\}$, the algorithm outputs
$z_t$ to be any such string.
Otherwise, we set $z_t$ to be an arbitrary string of $U$.

\begin{proof}[Proof of~\cref{lem:zero_algo}]
    By \cref{lem:zero_contain}, if the algorithm outputs a string from $S_{n(x_t)} \setminus \{x_1, \dots, x_t\}$,
    then the output $z_t$ must be a correct unseen string from the target language.
    Thus, the algorithm outputs an incorrect string at some time $t$
    only if $S_{n(x_t)} \subseteq \{x_1, \dots, x_t\}$.
    Let $F_t = \{t' \le t \mid S_{n(x_{t'})} \subseteq \{x_1, \dots, x_{t'}\}\}$
    be the set of such times up to $t$.
    We have $\abs{W_t} \le \abs{F_t}$, so it suffices to argue that
    $\limsup_{t \to \infty} \abs{F_t}/t = 0$.

    The adversary is not allowed to repeat strings in its enumeration,
    so for each node $v$, there can be at most one time $t$
    such that $n(x_t) = v$ and $S_{n(x_t)} \subseteq \{x_1, \dots, x_t\}$.
    To bound $\abs{F_t}$, note that if $S_{n(x_t)} \subseteq \{x_1, \dots, x_t\}$
    at some time $t$, the adversary must have output exactly $\abs{S_{n(x_t)}}$
    strings from $S_{n(x_t)}$ in the first $t$ times.
    For any ray $P = (v_1, v_2, \dots) \in I$, the $m$ smallest sets
    are exactly the sets $S_{v_1}$, \dots, $S_{v_m}$
    which have a total size of $1 + \dots + m = \Omega(m^2)$.
    This implies that at time $t$, the adversary can only fully output the strings
    from at most $O(\sqrt{t})$ sets, so $\abs{F_t} = O(\sqrt{t})$.
    Thus,
    \[\limsup_{t \to \infty} \frac{\abs{W_t}}{t} \le \limsup_{t \to \infty} \frac{\abs{F_t}}{t} = 0\]
    as desired.
\end{proof}

We now show that there does not exist an algorithm that generates in the limit for $\SC_{0}$.

\begin{lemma}
    \label{lem:zero_nexist}
    There does not exist an algorithm that generates in the limit for $\SC_0$.
\end{lemma}
\begin{proof}
    Assume for contradiction that there exists an algorithm $G$ that generates in the limit for $\SC_0$.
    We will inductively define an adversary strategy and corresponding target language $K \in \SC_{0}$
    which ensures that $G$ does not generate in the limit.

    The adversary strategy will proceed in an infinite number of stages, where at each stage, the
    adversary ensures that the algorithm generates at least one incorrect string.
    At any stage $i$, the adversary will completely enumerate the strings in $S_v$
    for some node $v$ that is at depth $i$.
    At the time $t$ when the adversary outputs the last string in $S_v$,
    we will argue that we can choose $K$ so that the algorithm's output $z_t$
    is not an unseen string from $K$.

    At some stage $i$, inductively assume that the adversary has chosen some path
    of nodes $v_1$, \dots, $v_{i}$ such that $v_1 = \emptyset$ is the root,
    and $v_{j}$ is the child of $v_{j-1}$ for all $j \ge 2$.
    For the base case of $i = 1$, set $v_1 = \emptyset$.
    Also assume that inductively, the adversary has output every string
    from $S_{v_1} \cup \dots \cup S_{v_{i-1}}$ by the time stage $i$ begins.
    During stage $i$, we will output all of the strings in $S_{v_i}$.
    Note that $v_i$ must have depth $i$, so $\abs{S_{v_i}} = i$.
    For each of the next $i$ time-steps, the adversary will output a distinct string from $S_{v_i}$
    in some arbitrary order.
    Let $t_i$ be the last of those $i$ time-steps, i.e., the time when the adversary
    outputs the last string from $S_{v_i}$.

    If $z_{t_i}$ is not contained in the subtree of $v_i$,
    we set $v_{i+1}$ to be an arbitrary child of $v_i$.
    Otherwise, since $v_i$ has two children, there must be a child of $v_i$
    whose subtree does not contain the algorithm's output $z_{t_i}$.
    We set $v_{i+1}$ to be the child of $v_i$ that does not contain $z_{t_i}$.

    Finally, the target language $K = \{x_1, x_2, \dots\}$ is the set of all strings
    output by the adversary.
    Clearly $K \in \SC_0$ because the adversary chooses $v_{i+1}$
    to be a child of $v_i$ at each stage, and then enumerates every string of $S_{v_{i+1}}$.
    We now claim that for each $i \in \NN$, the algorithm's output $z_{t_i}$ is not an unseen
    string from $K$.
    If $z_{t_i}$ is not contained in the subtree of $v_i$, then it must either
    have already been output by the adversary, or does not appear in $K$.
    If $z_{t_i}$ is contained in $S_{v_i}$, it must have already been output by the adversary.
    Otherwise, if $z_{t_i}$ is contained in the subtree of a child of $v_i$,
    we chose $v_{i+1}$ so that $z_{t_i}$ is not contained in the subtree of $v_{i+1}$.
    At any future time, the adversary will only output strings that are contained
    in the subtree of $v_{i+1}$, implying that $z_{t_i} \notin K$ as desired.

    We have shown that for any algorithm $G$, there exists a target language $K \in \SC_0$
    and an adversary enumeration $x_1$, $x_2$, \dots\ of $K$ such that $G$
    generates incorrectly infinitely often.
    Thus, there does not exist an algorithm that generates in the limit for $\SC_{0}$.
\end{proof}

Combining \cref{lem:zero_algo,lem:zero_nexist} gives the proof for \cref{thm:zero}.

\subsection{Hierarchy of Generation}

We now show that there exists a collection $\SC$ that is not generatable
with any hallucination rate less than $1$.
In fact, for any algorithm $G$, there exists an adversary strategy such that $G$
never generates a single correct string.
Compared to \cref{thm:zero}, this forms the other end of the hierarchy for generation with infinite hallucination.
The proof closely follows the proof of Proposition~A.1 in~\cite{HKMV25}.

\ungenall*

\begin{proof}
    Let the collection $\SC = \{L \subseteq \NN \mid \abs{L} = \infty\}$
    be the set of all infinite subsets of $\NN$
    and fix an arbitrary algorithm $G$.
    We will show that for any $G$, there exists a target language $K \in \SC$ and adversary enumeration
    such that $G$ never generates a single correct string.
    This would imply that $\SC$ is not $\beta$-generatable in the limit for any $\beta < 1$.

    We inductively construct a sequence of adversary outputs $x_1$, $x_2$, \dots.
    For the base case, let $x_1 = 1$.
    Then, for any $t > 1$, set $x_t = 1 + \max\{x_1, z_1, \dots, x_{t-1}, z_{t-1}\}$.
    Finally, let the target language be $K = \{x_1, x_2, \dots\}$.
    Clearly $K$ is an infinite subset of $\NN$, so $K \in \SC$.

    To see that $G$ never generates a correct unseen string, let $t$ be an arbitrary time,
    and consider the algorithm's output $z_{t}$.
    The output $z_t$ is a correct unseen string from $K$ if and only if $z_t \in \{x_{t+1}, x_{t+2}, \dots\}$.
    However, at any future time $t' > t$, the adversary's output $x_{t'}$ is set to be strictly larger than $z_t$.
    Thus, $G$ never generates a single correct string, implying that
    $\SC$ is not $\beta$-generatable in the limit for any $\beta < 1$.
\end{proof}

Our next result establishes a strict hierarchy of generation at every rate of hallucination.

\real*

The collection $\SC_{\beta}$ is constructed in a similar manner to the collection $\SC_0$
used in the proof of \cref{thm:zero}.
As before, we will have an infinite rooted tree $T_{\beta}$ where each node $v$ contains a pool of strings.
To define the tree, we first need to introduce two sequences $\{p_{i}\}, \{q_{i}\}$ for $i \ge 1$.
Each node in a particular layer $i$ of the tree will have exactly $q_{i} + 1$ children.
We fix $\{p_{i}\}$ and $\{q_{i}\}$ to be any sequences of positive integers
such that $p_{i} / q_{i}$ converges to $\beta$ from below,
i.e., satisfying the properties that $1 \le p_{i} \le q_{i}$ and $p_{i}/q_{i} \le \beta$ for all $i \in \NN$,
and $\lim_{i \to \infty} p_{i}/q_{i} = \beta$.
It is easy to see that such sequences can be constructed for any $\beta$.
Note that the sequences $\{p_i\}$ and $\{q_i\}$ depend on $\beta$, but we omit
$\beta$ from the subscripts for ease of notation.

Since each node at depth $i$ has exactly $q_{i} + 1$ children,
a node at depth $i$ can be labeled by $i-1$ tuples where the $j$-th coordinate takes values from $[q_{j} + 1]$.
The root of $T_{\beta}$ is labeled by $\emptyset$, and we inductively assign
$(u_1, \dots, u_{i})$ to label the $u_{i}$-th child of $(u_1, \dots, u_{i-1})$.
As before, we can think of the label of a node $v$ as encoding the unique path from the root to $v$.
The total set of labels for nodes in $T_{\beta}$ is
\[\bigcup_{i \in \NN} \prod_{j < i} [q_{j} + 1].\]

For each node $v \in T_{\beta}$ at some depth $i = d(v)$, we associate with $v$ a set
\[D_v = \{vj \mid j \in [3q_{i} - p_{i} + 1]\}.\]
The set $D_v$ is further partitioned into two sets $A_v = \{vj \mid j \in [2q_{i}]\}$ of size $2q_{i}$,
and $B_v = \{vj \mid j \in \{2q_{i} + 1, \dots, 3q_{i} - p_{i} + 1\}\}$ of size $q_{i} - p_{i} + 1$.
Recall that $v$ is an $i - 1$ tuple, and $vj$ is the $i$ tuple formed by concatenating $v$ with the integer $j$.
\Cref{fig:real-tree} depicts an example of a possible tree
for specific values of $\{p_{i}\}$ and $\{q_{i}\}$.

\begin{figure}[t]
    \centering
    \scalebox{.85}{
        \begin{tikzpicture}[
            scale=0.84,
            transform shape,
            edge/.style={draw, very thick},
            info/.style={
                draw,
                rounded corners=10pt,
                thick,
                fill=gray!10,
                align=center,
                inner sep=6pt,
                font=\small
            },
            info3/.style={
                draw,
                rounded corners=10pt,
                thick,
                fill=gray!10,
                align=center,
                inner sep=4pt,
                font=\scriptsize
            },
            placeholder/.style={
                draw,
                rounded corners=10pt,
                thick,
                fill=gray!10,
                align=center,
                inner sep=4pt,
                minimum width=2.7cm,
                minimum height=1.8cm,
                font=\Large
            },
            dots/.style={font=\Large}
        ]

            \node[info, text width=3.5cm] (root) at (0,0)
                {$\begin{array}{c}
                      v=\emptyset\\[1mm]
                      A_v=\{(1),(2)\}\\[1mm]
                      B_v=\{(3)\}
            \end{array}$};

            \node[info, text width=4.9cm] (left) at (-5.0,-3.3)
                {$\begin{array}{c}
                      v=(1)\\[1mm]
                      A_v=\{(1,1),(1,2),(1,3),(1,4)\}\\[1mm]
                      B_v=\{(1,5),(1,6)\}
            \end{array}$};

            \node[info, text width=4.9cm] (right) at (5.0,-3.3)
                {$\begin{array}{c}
                      v=(2)\\[1mm]
                      A_v=\{(2,1),(2,2),(2,3),(2,4)\}\\[1mm]
                      B_v=\{(2,5),(2,6)\}
            \end{array}$};

            \node[info3, text width=3.2cm] (l1) at (-9.8,-7)
                {$\begin{array}{c}
                      v=(1,1)\\[0.7mm]
                      A_v=\{(1,1,1),(1,1,2),\\
                      (1,1,3),(1,1,4)\}\\[0.7mm]
                      B_v=\{(1,1,5)\}
            \end{array}$};

            \node[info3, text width=3.2cm] (l2) at (-6.0,-7)
                {$\begin{array}{c}
                      v=(1,2)\\[0.7mm]
                      A_v=\{(1,2,1),(1,2,2),\\
                      (1,2,3),(1,2,4)\}\\[0.7mm]
                      B_v=\{(1,2,5)\}
            \end{array}$};

            \node[info3, text width=3.2cm] (l3) at (-2.2,-7)
                {$\begin{array}{c}
                      v=(1,3)\\[0.7mm]
                      A_v=\{(1,3,1),(1,3,2),\\
                      (1,3,3),(1,3,4)\}\\[0.7mm]
                      B_v=\{(1,3,5)\}
            \end{array}$};

            \node[info3, text width=3.2cm] (r1) at (1.8,-7)
                {$\begin{array}{c}
                      v=(2,1)\\[0.7mm]
                      A_v=\{(2,1,1),(2,1,2),\\
                      (2,1,3),(2,1,4)\}\\[0.7mm]
                      B_v=\{(2,1,5)\}
            \end{array}$};

            \node[placeholder] (r2) at (5.6,-7) {$\vdots$};
            \node[placeholder] (r3) at (9.4,-7) {$\vdots$};

            \draw[edge] (root.south) -- (left.north);
            \draw[edge] (root.south) -- (right.north);

            \draw[edge] (left.south) -- (l1.north);
            \draw[edge] (left.south)      -- (l2.north);
            \draw[edge] (left.south) -- (l3.north);

            \draw[edge] (right.south) -- (r1.north);
            \draw[edge] (right.south)      -- (r2.north);
            \draw[edge] (right.south) -- (r3.north);

            \node[dots] at (-9.8,-8.6) {$\vdots$};
            \node[dots] at (-6.0,-8.6) {$\vdots$};
            \node[dots] at (-2.2,-8.6) {$\vdots$};

            \node[dots] at (1.8,-8.6) {$\vdots$};
            \node[dots] at (5.6,-8.6) {$\vdots$};
            \node[dots] at (9.4,-8.6) {$\vdots$};
        \end{tikzpicture}
    }
    \caption{A partial view of the rooted tree from the proof of \cref{thm:real},
        with $(p_{1},q_{1}) = (1,1)$ at the root,
        $(p_{2}, q_{2}) = (1,2)$ at the second layer,
        and $(p_{3}, q_{3}) = (2,2)$ at the third layer.
        Note that these values of $p_{i}$ and $q_{i}$
        are chosen for illustrative purposes and do not correspond to the
        construction of $p_{i}$ and $q_{i}$ used in the proof of \cref{thm:real}.
        For each node $v$, $D_v = A_v \cup B_v$.
    }
    \label{fig:real-tree}
\end{figure}
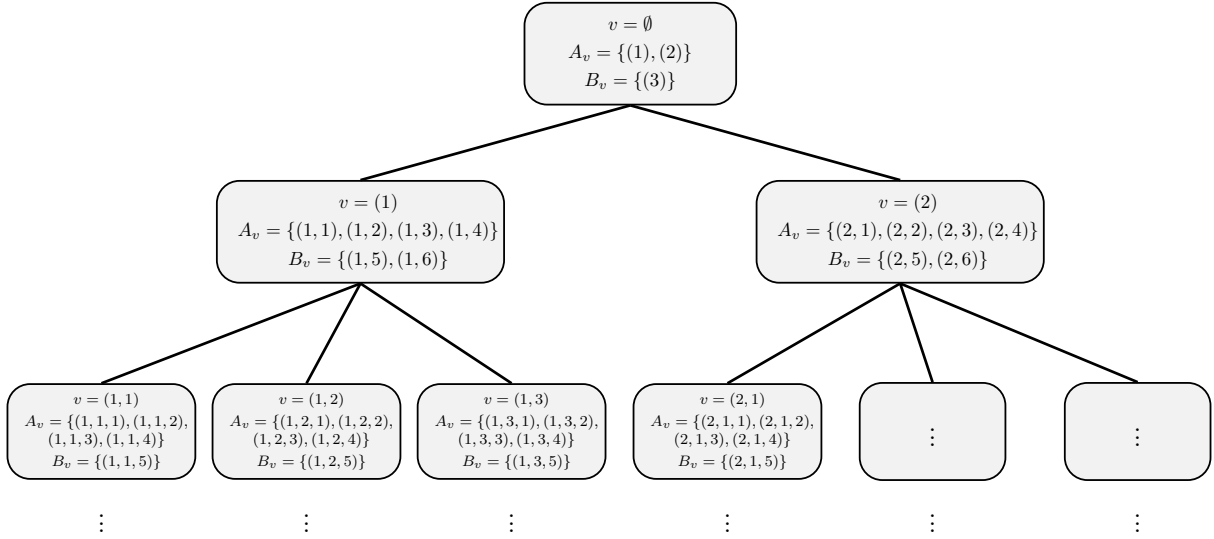

We now describe the collection $\SC_{\beta}$.
The universe $U_{\beta} = \bigcup_{v \in T_{\beta}} D_v$ is simply the union of the sets at each node in the tree.
It is clear that $U_{\beta}$ is countable.
A language $L \in \SC_{\beta}$ will be partially parameterized by a ray in $T_{\beta}$.
Let $I_{\beta} = \{(v_1, v_2, v_3, \dots) \mid v_1 = \emptyset,
\text{ and $v_j$ is a child of $v_{j-1}$ for all $j \ge 2$}\}$
be the set of all infinite rays in $T_{\beta}$ that originate at the root.
Each ray $P \in I_{\beta}$ will now be associated with an infinite set of languages $E_P$ that we will describe now.
Each language $L \in E_P$, will contain only strings from $D_{v_i}$ where $v_i$ is contained in the ray $P$.
In particular, each language $L$ will contain all of the strings from each $B_{v_i}$ along with exactly $p_{i} - 1$
strings from each $A_{v_i}$.
The set $E_P$ consists of the languages formed from all possible ways of choosing exactly $p_{i} - 1$ strings from each
$A_{v_i}$.
Note that in total, each language contains exactly $q_i$ strings from each layer $D_{v_i}$---exactly
$q_{i} - p_{i} + 1$ strings from $B_{v_i}$, and $p_{i} - 1$ strings from $A_{v_i}$.
The final collection $\SC_{\beta} = \bigcup_{P \in I_{\beta}} E_P$
is simply the union of the sets of languages for each ray.

In a rough sense, the algorithm tries to identify the ray $P$ corresponding to the target language.
If the algorithm identifies that some node $v$ is part of the ray
(because the adversary output a string from the corresponding set $D_v$), then the algorithm can
safely generate from the entire set $B_v$ of size $q_i - p_i + 1$.
However, the adversary can delay revealing new nodes by outputting $q_i$ strings
at each layer $i$, before moving on to the next layer.
So in a very rough sense, the algorithm can generate correctly for $q_i - p_i$ out of $q_i$ time-steps
at each layer $i$, but can do no better.

\begin{lemma}
    \label{lem:real_algo}
    There exists an algorithm that $\beta$-generates in the limit for $\SC_{\beta}$.
\end{lemma}

The algorithm follows a similar idea as the algorithm in \cref{lem:zero_algo},
where the algorithm attempts to identify a set of strings that are safe to generate from.
Since a language $L_{P} \in \SC_{\beta}$ always contains every string in $B_v$
for nodes $v \in P$, if the adversary ever reveals a string from some set $D_v = A_v \cup B_v$,
then the entire set $B_v$ must be contained in the target language.
The set $B_v$ has size $q_{d(v)} - p_{d(v)} + 1$, implying that the algorithm
can generate correct strings from $B_v$ for at least $q_{d(v)} - p_{d(v)}$ time-steps.
In total, the target language contains exactly $q_{d(v)}$ strings from $D_v$,
allowing us to show that the algorithm generates incorrectly at most $p_{d(v)} / q_{d(v)}$
fraction of the time from $D_v$.
Since $p_{i} / q_{i} \le \beta$ for all $i$,
this implies that the algorithm $\beta$-generates in the limit.

For any string $x \in U_{\beta}$, we define $n(x)$ to be the unique node $v \in T_{\beta}$
such that $x \in D_v$.

\begin{lemma}
    \label{lem:real_contain}
    For any language $L \in \SC_{\beta}$ and string $x$, if $x \in L$, then $B_{n(x)} \subseteq L$.
\end{lemma}
\begin{proof}
    Let $P$ be the infinite ray that corresponds to the language $L$.
    Since $L$ can only contain strings from sets $D_v$ where $v \in P$,
    it must be that $n(x) \in P$.
    Because $L$ contains the set $B_v$ for every $v \in P$,
    this implies that $B_{n(x)} \subseteq L$.
\end{proof}

We now formally define the algorithm that $\beta$-generates in the limit for $\SC_{\beta}$.
At some time $t$, if there exists a string in $B_{n(x_t)} \setminus \{x_1, \dots, x_t\}$, the algorithm outputs
$z_t$ to be any such string.
Otherwise, we set $z_t$ to be an arbitrary string of $U_{\beta}$.
By \cref{lem:real_contain}, if the algorithm outputs a string from $B_{n(x_t)} \setminus \{x_1, \dots, x_t\}$
at some time $t$, then the output $z_t$ must be a correct unseen string from the target language.
Thus, the algorithm outputs an incorrect string at some time $t$
only if 
\begin{equation}
    B_{n(x_t)} \subseteq \{x_1, \dots, x_t\}.  \label{eq:tbad}
\end{equation}
We say $t$ is bad if (\ref{eq:tbad}) holds. 
For any node $v \in T_{\beta}$, let $F_{t}(v) = \{t' \text{ bad} \mid x_{t'} \in D_v\}$
be the set of bad $t$ when the adversary also outputs a string from $D_v$.
Also let $H_t(v) = \{t' \le t \mid x_{t'} \in D_v\}$ be the set of times
at which the adversary outputs any string from $D_v$.
Since the algorithm only outputs incorrectly at times $t$ bad, the upper density of wrong output timestamps is at most
\begin{equation}
    \label{eqn:real_bound}
     \limsup_{t \to \infty} \frac{\sum_{v \in T_{\beta}} \abs{F_t(v)}}{\sum_{v \in T_{\beta}} \abs{H_t(v)}}.
\end{equation}
Thus, it suffices to bound the value of $\abs{F_t(v)} / \abs{H_t(v)}$ for every $v$.

\begin{lemma}
    \label{lem:real_vertex_bound}
    For every time $t$ and $v \in T_{\beta}$, either $\abs{F_t(v)} = 0$
    or $\abs{F_t(v)} / \abs{H_t(v)} \le \beta$.
\end{lemma}
\begin{proof}
    Let $v$ be a node at depth $i$.
    Intuitively, the algorithm can generate correctly at least the first $\abs{B_v} - 1$
    times the adversary outputs from $D_v$, implying that the algorithm generates incorrectly at most
    $(q_i - (\abs{B_v} - 1))/q_i = p_{i}/q_{i} \le \beta$ fraction of the time from $D_v$.

    Assume that $\abs{F_t(v)} \ge 1$.
    This implies that the adversary has already enumerated every string in $B_v$.
    For each of the first $\abs{B_v} - 1$ times that the adversary outputs a string from $D_v$,
    it cannot have been that $B_{n(x_{t'})} \subseteq \{x_1, \dots, x_{t'}\}$.
    Thus, $\abs{H_t(v)} \ge \abs{F_t(v)} + \abs{B_v} - 1 = \abs{F_t(v)} + q_{i} - p_{i}$.
    Since every language in $\SC_{\beta}$ contains at most $q_i$ strings from $D_v$,
    we also have $\abs{H_t(v)} \le q_i$.
    Thus, 
    $\abs{F_t(v)} / \abs{H_t(v)}
    \leq p_{i} / q_{i} \le \beta$.
\end{proof}

\begin{proof}[Proof of~\cref{lem:real_algo}]
    Combining \cref{eqn:real_bound,lem:real_vertex_bound},
    we see that $\limsup_{t \to \infty} \abs{W_t}/t \le \beta$ as desired.
    In fact, we note that our proof shows the stronger statement that
    $\sup_{t \to \infty} \abs{W_t}/t \le \beta$.
\end{proof}

We now prove the other direction of \cref{thm:real}.

\begin{lemma}
    \label{lem:real_nexist}
    There does not exist an algorithm that $\beta'$-generates in the limit for $C_{\beta}$
    for any $\beta' < \beta$.
\end{lemma}

\begin{proof}
    Fix an arbitrary $\beta'< \beta$ and assume for contradiction that
    there exists an algorithm $G$ that $\beta'$-generates in the limit for $\SC_{\beta}$.
    We will inductively define an adversary strategy and corresponding target language $K \in \SC_{\beta}$
    which ensures that $G$ does not $\beta'$-generate in the limit for any $\beta' < \beta$.
    The adversary construction will follow a similar idea to the proof of \cref{lem:zero_nexist},
    where we inductively build a ray starting from the root.

    At some stage $i$, inductively assume that the adversary has chosen some path
    of nodes $v_1$, \dots, $v_{i}$ such that $v_1 = \emptyset$ is the root,
    and $v_{j}$ is the child of $v_{j-1}$ for all $j \ge 2$.
    For the base case of $i = 1$, set $v_1 = \emptyset$.
    Also assume that inductively, the adversary has output every string from $B_{v_j}$
    and exactly $p_{j} - 1$ strings from $A_{v_j}$ for every $j < i$
    at the beginning of stage $i$.
    Finally, also inductively assume that by the beginning of stage $i$,
    the algorithm has output no strings from the subtree of $v_i$.

    During stage $i$, the adversary will output every string from $B_{v_i}$
    and $p_{i} - 1$ strings from $A_{v_i}$.
    Let stage $i$ start at time $t_i$.
    For the first $q_{i} - p_{i} + 1$ time-steps
    when $t \in [t_i, t_i + q_{i} - p_{i}]$,
    the adversary enumerates the strings of $B_{v_i}$ in arbitrary order.
    Then, for each of the next $p_{i} - 1$ time-steps
    when $t \in [t_i + q_{i} - p_{i} + 1, t_i + q_{i} - 1]$,
    the adversary outputs an arbitrary unrevealed string from $A_{v_i}$.
    To see that such an unrevealed string exists at each time,
    note that the adversary outputs $q_{i}$ total strings in stage $i$.
    Thus, the adversary and algorithm collectively output at most $2q_{i}$
    strings during stage $i$.
    Since the algorithm has not output any strings from $D_{v_i}$ before the start of stage $i$,
    and we have $\abs{A_{v_i}} = 2q_{i}$, this implies that an unrevealed string from $A_{v_i}$
    exists at each time.

    Finally, we choose $v_{i+1}$ to be a child of $v_i$ such that the subtree of $v_{i+1}$
    does not contain any strings output by the algorithm up to the end of stage $i$.
    Since we assumed inductively that the algorithm did not output any strings from the subtree of $v_i$ before the beginning of stage $i$, the algorithm can output at most $q_i$ strings from the subtree of $v_i$ by the end of stage $i$.
    Since $v_i$ has $q_{i} + 1$ children, there must be a child of $v_i$
    that does not contain any strings output by the algorithm, fulfilling the inductive hypothesis.

    We set $K = \{x_1, x_2, \dots\}$.
    Note that the adversary strategy constructs an infinite ray $v_1$, $v_2$, \dots\ where
    for each $v_{i}$, the language $K$ contains all of $B_{v_i}$ along with exactly $p_{i} - 1$
    strings from $A_{v_i}$.
    The collection $\SC_{\beta}$ exactly contains all languages of the above form, so $K \in \SC_{\beta}$.

    We now argue that in each stage $i$, at least $p_{i}$ strings output by the algorithm are incorrect.
    In particular, for each $t' \in [t_i + q_{i} - p_{i}, t_i + q_{i} - 1]$,
    the output $z_{t'}$ must be incorrect.
    Let $t'$ be an arbitrary time in $[t_i + q_{i} - p_{i}, t_i + q_{i} - 1]$.
    If $z_{t'}$ is not contained in the subtree of $v_i$, then it must either
    have already been output by the adversary, or does not appear in $K$.
    If $z_{t'} \in B_{v_i}$, then note that the adversary has already output every string from $B_{v_i}$
    by time $t_i + q_{i} - p_{i}$.
    Thus $z_{t'}$ cannot be a correct unseen string from $K$.

    If $z_{t'} \in A_{v_i}$, the adversary is constructed so that it never outputs $z_{t'}$
    at any future time.
    Thus $z_{t'}$ cannot be a correct unseen string from $K$.
    Finally, if $z_{t'}$ is contained in the subtree of some child of $v_{i}$,
    the adversary chooses $v_{i+1}$ so that no algorithm outputs are contained in $v_{i+1}$.
    This again implies $z_{t'} \notin K$.
    Thus, for each $t' \in [t_i + q_{i} - p_{i}, t_i + q_{i} - 1]$
    the algorithm's output $z_{t'}$ cannot be a correct unseen string from $K$.
    Since the algorithm outputs $q_{i}$ strings in total during stage $i$, we have
    \[\limsup_{t \to \infty} \frac{\abs{W_t}}{t} \ge
    \limsup_{i \to \infty} \frac{\sum_{j < i} p_{j}}{\sum_{j < i} q_{j}}.\]
    The sequence $p_{i}/q_{i}$ converges to $\beta$ from below,
    implying that $(\sum_{j < i} p_{j})/(\sum_{j < i} q_{j})$
    also converges to $\beta$.
    Thus, $\limsup_{t \to \infty} \abs{W_t}/t \ge \beta$, so $G$
    does not $\beta'$-generate in the limit for any $\beta' < \beta$.
\end{proof}

\subsection{Hierarchy of Generation with Breadth}

We now explore the tradeoff between the rate of hallucination and the breadth of language generation.
Our first result shows that there is a collection where allowing arbitrary amounts of hallucination
does not improve the breadth of generation.

\densityzero*
\begin{proof}
    Let $A = \{2^i \mid i \in \NN\}$ and the uncountable collection of languages $\SC$ be all the supersets of $A$ in $\mathbb{N}$.
    Clearly $\SC$ can be generated in the limit with finite error (in fact, with no errors at all)
    because the algorithm can simply output an unrevealed string from $A$ at each time-step.
    Now fix an arbitrary $\beta \le 1$ and let $G$ be an algorithm
    that $\beta$-generates in the limit.
    Consider the following inductive adversary strategy.
    When $t$ is odd, output the smallest unrevealed string in $\NN \setminus A$,
    and when $t$ is even, output the smallest string in $A$
    not yet output by the adversary.
    Finally, let the target language simply be $K = \{x_1, x_2, \dots\}$.

    Let the set of strings output by the algorithm be $O = \{z_1, z_2, \dots\}$.
    It is clear that $K = (\NN \setminus O) \cup A$, so $K \in \SC$.
    Additionally, $(O \cap K) \subseteq A$,
    so it suffices to show that $A$ has upper density $0$ in $K$.
    We claim that $K$ has lower density at least $1/4$ in $\NN \setminus A$.
    This is because at every odd time $t$, the adversary outputs
    the smallest string of $\NN \setminus A$ not output by the adversary or algorithm.
    Between subsequent odd times $t$ and $t+2$, there are at most $3$ strings
    output by the adversary or algorithm -- $z_{t}$, $x_{t+1}$, and $z_{t+1}$.
    This implies that the set $\{x_1, x_3, x_5, \dots\}$ has lower density at least $1/4$
    in $\NN \setminus A$.
    Since $K$ has lower density at least $1/4$ in $\NN \setminus A$,
    and $A$ has density $0$ in $\NN$, this implies that $A$ has upper density $0$ in $K$.
    Thus, $G$ generates with upper density at most $0$.
\end{proof}

We next establish a strict hierarchy of generation at every rate of hallucination
and breadth of generation.

\densityhierarchy*

Fix an arbitrary $\beta \in (0, 1]$ and $\alpha \in [0, 1/2]$.
The collection $\SC_{\beta, \alpha}$ will be very similar to the collection $\SC_{\beta}$
used in the proof of \cref{thm:real}.
We again define $\{p_i\}$ and $\{q_i\}$ to be any two sequences that converge to $\beta$
from below, with the additional restriction that $q_i - p_i \ge 2$ for all $i \in \NN$.
The tree $T_{\beta}$ is defined identically to the tree in the proof of \cref{thm:real}.
At each node $v$ of depth $i$, the set $B_v$ that is common to all languages passing through $v$
is still defined to be a set of $q_{i} - p_{i} + 1$ strings.
For notational ease, we will write $B_v = \{(v, j) \mid j \in [q_{i} - p_{i} + 1]\}$.
The only true difference comes in the definition of the set $A_v$.
For some node $v$ at depth $i$, we will think of $A_v$ as containing $p_i - 1$ distinct bins,
each containing exactly $q_i + 1$ strings.
Formally, we have $A_v = \{(v, a, b) \mid a \in [p_i - 1], b \in [q_i + 1]\}$
and the $j$-th bin of $A_v$ is the set $A_{v, j} = \{(v, j, b) \mid b \in [q_i + 1]\}$.
We then let $D_v = A_v \cup B_v$, and $U = \bigcup_{v \in T_{\beta}} D_v$.

Each ray $P \in I_{\beta}$ will again correspond to a set of languages $E_P$.
As before, each language $L \in E_P$ will contain only strings from $D_{v_i}$ where $v_i$
is contained in the ray $P$.
In particular, each language $L$ will contain all of the strings from each $B_{v_i}$ along with exactly $p_{i} - 1$
strings from each $A_{v_i}$, where we choose exactly one string from each bin of $A_{v_i}$.
The set $E_P$ consists of the languages formed from all possible ways of
choosing exactly one string from each bin of $A_{v_i}$.
As before, each language contains exactly $q_i$ strings from each layer $D_{v_i}$---exactly
$q_{i} - p_{i} + 1$ strings from $B_{v_i}$, and $p_{i} - 1$ strings from $A_{v_i}$.
The final collection $\SC_{\beta, \alpha} = \bigcup_{P \in I_{\beta}} E_P$
is simply the union of the sets of languages for each ray.

Now recall that the notion of breadth of generation depends on an ordering of the strings of the universe $U$.
We define the linear ordering of $U$ as follows. 
First partition $\mathbb{N}$ into three sets $H$, $R$, and $Y$
such that $H$ has upper and lower density $2\alpha$, $R$ has upper and lower density $0$, and $Y$ has upper and lower density $1 - 2 \alpha$.
We label the elements in each set in increasing order,
i.e., $H = \{h_1 < h_2 < \dotsm\}$, $R = \{r_1 < r_2 < \dotsm\}$, and $Y = \{y_1 < y_2 < \dotsm\}$.

We define the linear ordering of $U$ via a (non-injective) mapping $f \colon U \to \NN$, where we say for
strings $s, t \in U$ that $s \prec t$ if $f(s) < f(t)$.
Two strings with $f(s) = f(t)$ are incomparable, but our proof will only ever compare strings where $f(s) \neq f(t)$.
At every node $v$, the function $f$ will map exactly two strings from $B_{v}$ into the density $2 \alpha$ set $H$.
The rest of the strings from $B_v$ will be mapped into the density $0$ set $R$.
Finally, the strings from $A_v$ will be mapped into the set $Y$ while ensuring that all strings
within the same box are mapped to the same element of $Y$.
The mapping is layer-order-preserving in $H$, and similarly in $R$ and $Y$:
for any $i<i'$, every image of a string in layer $i$ precedes every image of a string in layer $i'$ in $H$.
Furthermore, there are no gaps in that the image of the strings at a node $v$
immediately follow the image of the strings at $v$'s parent.

Intuitively, this mapping will guarantee that the algorithm can correctly output exactly one out of the two strings
mapped into $H$ at each level, ensuring that the algorithm generates with lower density $2\alpha/2 = \alpha$.
However, the algorithm can do no better because it cannot generate any correct strings from the set $A_{v}$,
and the remaining strings in $B_v$ that are mapped to $R$ do not contribute to the lower density of generation.
The argument that $\SC_{\beta, \alpha}$ can be $\beta$-generated, but not $\beta'$-generated for any $\beta' < \beta$
follows the same idea as in \cref{thm:real}.

We first show that there exists an algorithm that $\beta$-generates in the limit with lower density $\alpha$.

\begin{lemma}
    \label{lem:density_algo}
    There exists an algorithm that $\beta$-generates in the limit with lower density $\alpha$ for $\SC_{\beta, \alpha}$.
\end{lemma}
\begin{proof}
    The algorithm will be almost identical to the algorithm described in the proof of \cref{lem:real_algo},
    except we first choose to output strings that map to elements in $H$.
    Recall that for any string $x \in U_{\beta}$,
    we define $n(x)$ to be the unique node $v \in T$ such that $x \in D_v$.
    At some time $t$, if there exists a string $y$ in $B_{n(x_t)} \setminus \{x_1, \dots, x_t\}$
    such that $f(y) \in H$, the algorithm outputs $z_t$ to be any such string.
    Otherwise, we try setting $z_t$ to be any string in $B_{n(x_t)} \setminus \{x_1, \dots, x_t\}$.
    If $B_{n(x_t)} \subseteq \{x_1, \dots, x_t\}$, we set $z_t$ to be an arbitrary string of $U_{\beta}$.

    Once an adversary outputs a string from some set $D_v$, it is still true that
    the entire set $B_v$ must be contained in the target language.
    Thus, the argument in the proof of \cref{lem:real_algo} directly shows
    that the algorithm $\beta$-generates in the limit.

    We now show that the algorithm generates with lower density $\alpha$.
    Let $K$ be the target language and $P = (v_1, v_2, \dots)$ be the ray corresponding to $K$.
    Also let $O = \{z_1, z_2, \dots\}$ be the total set of algorithm outputs,
    and $f(O)$ and $f(K)$ be the images of $O$ and $K$ under $f$.
    The lower density of the algorithm
    is equal to the lower density of $f(O)$ in $f(K)$.
    By the property that the image of strings at a node $v$ immediately follow the image of strings at $v$'s parent, we have that $f(K) = \NN$.
    Thus, we wish to find the lower density of $f(O)$ in $\NN$.
    Since each set $B_{v_i}$ contains two strings that map to $H$, and the adversary
    must eventually enumerate a string from each $D_{v_i}$,
    the algorithm's set of outputs $O$ is guaranteed to contain at least one string
    that maps to an element of $H$ from each set $B_{v_i}$.
    Because $H$ has lower density $2 \alpha$ in $\NN$, this implies $f(O)$ has lower density at least $\alpha$ in $\NN$,
    implying that the algorithm generates with lower density $\alpha$.
\end{proof}

\begin{lemma}
    \label{lem:density_beta_nexist}
    There does not exist an algorithm that $\beta'$-generates in the limit for $C_{\beta, \alpha}$
    for any $\beta' < \beta$ (regardless of the breadth of generation).
\end{lemma}
\begin{proof}
    The proof is very similar to the proof of \cref{lem:real_nexist}.
    Let $G$ be any algorithm.
    We will inductively define an adversary strategy and corresponding target language $K \in \SC_{\beta, \alpha}$
    which ensures that $G$ does not $\beta'$-generate in the limit for any $\beta' < \beta$.

    Identically to \cref{lem:real_nexist}, at some stage $i$, inductively assume that the adversary has chosen some path
    $v_1$, \dots, $v_{i}$ starting from the root.
    Let stage $i$ start at time $t_i$.
    For the first $q_{i} - p_{i} + 1$ time-steps when $t \in [t_i, t_i + q_{i} - p_{i}]$,
    the adversary enumerates the strings of $B_{v_i}$ in order.
    Now, unlike in \cref{lem:real_nexist}, we need to output exactly one string from each bin in $A_{v_i}$.
    We will simply go through the bins in order and output any string that has not been output
    by the algorithm.
    For each $t \in [t_i + q_{i} - p_{i} + 1, t_i + q_{i} - 1]$,
    let $j_t = t - (t_i + q_{i} - p_{i})$, i.e., $j_t$ starts at $1$ and goes up to $p_i - 1$.
    At each time $t$, we set $x_t$ to be an arbitrary string from the bin $A_{v_i, j_t}$
    that has not yet been output by the algorithm.
    Since each bin has size $q_i + 1$, and the algorithm can only output $q_i$
    strings during stage $i$, such a string must exist.
    Finally, we choose $v_{i+1}$ to be a child of $v_i$ such that the subtree of $v_{i+1}$
    does not contain any strings output by the algorithm up to the end of stage $i$.

    We set $K = \{x_1, x_2, \dots\}$.
    At each stage $i$, we output all of the strings from $B_{v_i}$ and
    exactly one string from each bin of $A_{v_i}$, so $K \in \SC_{\beta, \alpha}$.

    The proof that $\limsup_{t \to \infty} \abs{W_t}/t \ge \beta$ is identical to
    the proof in \cref{lem:real_nexist}; the algorithm only outputs strings correctly
    during the first $q_i - p_i$ time-steps out of the $q_i$ total time-steps at stage $i$.
    Thus the algorithm does not $\beta'$-generate in the limit for any $\beta' < \beta$.
\end{proof}

\begin{lemma}
    \label{lem:density_alpha_nexist}
    There does not exist an algorithm that $\gamma$-generates in the limit with upper density $\alpha'$
    for $C_{\beta, \alpha}$ for any $\alpha' > \alpha$ and $\gamma \le 1$.
\end{lemma}
\begin{proof}
    We claim that under the adversary strategy described in the proof of \cref{lem:density_beta_nexist},
    no algorithm can generate with upper density greater than $\alpha$.
    Let $O = \{z_1, z_2, \dots\}$ be the set of algorithm outputs.
    As in the proof of \cref{lem:density_algo}, the algorithm's upper density
    is equal to the upper density of $f(O)$ in $\NN$, so it suffices to show that $f(O)$
    has at most $\alpha$ upper density.
    Note that at each stage $i$, the adversary enumerates the strings of $B_{v_i}$ in order.
    In particular, the first string that the adversary enumerates from each set $D_{v_i}$
    is a string in $B_{v_i}$ that maps to $H$.
    This implies that the algorithm's set of outputs $O$ contains at most one string
    that maps to a string of $H$ from each set $B_{v_i}$.
    Furthermore, we argued in the proof of \cref{lem:density_beta_nexist} that the algorithm (regardless of its rate of hallucination)
    does not output any correct strings from $A_{v_i}$. 
    This implies that $f(O)$ does not contain any strings from the set $Y$.
    Together with the fact that $R$ has upper density $0$ and $H$ has upper density $2 \alpha$, it must be that $f(O)$
    has upper density at most $0 + (2\alpha)/2 = \alpha$ in $\NN$.
    This implies that no algorithm can generate with upper density greater than $\alpha$
    regardless of the hallucination rate.
\end{proof}

Combining \cref{lem:density_algo,lem:density_beta_nexist,lem:density_alpha_nexist}
gives the proof of \cref{thm:two_sided_hierarchy}.

    \section{Generation with Hallucination Without Repetitions}\label{sec:without_repeats}In this section, we consider the setting where the algorithm is not allowed to repeat strings.
Recall that this restriction on the algorithm allows us to directly bound the fraction of
incorrect strings generated by the algorithm in terms of its rate of hallucination.

\subsection{Hierarchy of Generation}

Similar to the setting of generation with algorithm repetitions,
we establish a strict hierarchy on language collections generatable without repetition
at every allowed rate of hallucination.

\boxreal*

We first describe the collection $\SC_{\beta}$ used in the proof of the above theorem.
The universe $U$ on which the collection $\SC_{\beta}$ will be defined is a subset of $\NN \times \NN$.
As in the proof of \cref{thm:real}, we define two sequences $\{p_{i}\}$ and $\{q_{i}\}$ such that $p_{i}/q_{i}$
converges to $\beta$ from below.
However, we will need the additional property that $q_{i}$ does not grow too quickly.
For any $\beta \in (0, 1]$, define $\{p_{i}\}$ and $\{q_{i}\}$
to be two sequences of positive integers where
\[q_{i} = \ceil{\frac{1}{\beta}} + \floor{\log_2(i)}, \text{ and }
p_{i} = \floor{\beta \cdot q_{i}}.\]
It is clear that $\{p_{i}\}$ and $\{q_{i}\}$ satisfy the following conditions.
\begin{enumerate}
    \item $p_{i}/q_{i} \le \beta$ for all $i \in \NN$,
    \item $1 \le p_{i} \le q_{i}$ for all $i \in \NN$,
    \item $\lim_{i \to \infty} p_{2i+2} / i = \lim_{i \to \infty} q_{2i+2}/i = 0$,
    \item $\lim_{i \to \infty} p_{i}/q_{i} = \beta$.
\end{enumerate}

The universe $U$ on which we define the collection $\SC_{\beta}$ is given
by an infinite number of ``boxes'' where the $i$-th box
\[S_i \coloneq \{(1, i), (2, i), \dots, (2q_{i}, i)\}\]
is simply a set of $2q_{i}$ strings that are disjoint from the strings in any other box.
The universe $U = \bigcup_{i \in \NN} S_i$ is the union of all the boxes,
and the collection $\SC_{\beta}$ consists of all languages that
contain exactly $2q_{i} - p_{i}$ strings from the $i$-th box.
Formally, we have
\[\SC_{\beta} = \{L \subseteq U \mid \abs{L \cap S_i} =
2q_{i} - p_{i}, \, \forall \, i \in \NN\}.\]

We now prove both directions of~\cref{thm:box}.

\begin{lemma}
    \label{lem:box_algo}
    There exists an algorithm that $\beta$-generates in the limit without repetitions for $\SC_{\beta}$.
\end{lemma}
\begin{proof}
    The algorithm will proceed in an infinite number of stages, where
    in each stage $i$, we pick a new box $b_i$ consisting only of unrevealed strings.
    Then, as long as an unrevealed string from $S_{b_i}$ exists, the algorithm
    will output any such unrevealed string.
    Since $S_{b_i}$ contains $2q_{b_i}$ strings, and the adversary
    can only enumerate one string at each time-step, the algorithm will be able to
    output at least $q_{b_i}$ strings from $S_{b_i}$.
    Now, note that for any language $L \in \SC_{\beta}$, there are at most $p_{b_i}$ strings
    of $S_{b_i}$ that are not contained in $L$.
    Thus, at most $p_{b_i}$ of the algorithm outputs in $S_{b_i}$ are incorrect.
    Since $p_{b_i}/q_{b_i} \le \beta$, the fraction of errors after each stage is bounded by $\beta$.
    Within a particular stage, the first $p_{b_i}$ algorithm outputs may be incorrect,
    temporarily raising the error rate above $\beta$.
    However, because $q_{i}/i = o(1)$, we can show that in the limit, the
    errors accumulated in the first $p_{b_i}$ steps of each stage
    do not contribute to the error rate, and hence the $\limsup$ of the error rate is bounded by $\beta$ as desired.

    Formally, let the $i$-th stage start at time $t_i + 1$,
    and let $b_i$
    be the first box that contains only unrevealed strings.
    For each of the next $q_{b_i}$ times $t \in [t_i + 1, t_i + q_{b_i}]$,
    we set the algorithm's output $z_t$
    to be any string of $S_{b_i}$ that has not been output by the adversary or algorithm at time $t$.
    Since $\abs{S_{b_i}} = 2q_{b_i}$, such a string $z_t$ will always exist.
    We now show that the algorithm $\beta$-generates in the limit.

    We first prove by induction that the fraction of errors after each stage is at most $\beta$, i.e.,
    $\abs{W_{t_i}} / t_i \le \beta$ for all $i$.
    Within any stage $i$, the algorithm outputs at least $q_{b_i}$ distinct strings from box $b_i$.
    Note that for any language $L \in \SC_{\beta}$, at most $p_{b_i}$ strings of box $b_i$
    are not contained in $L$, implying that $\abs{W_{t_{i+1}}} \le \abs{W_{t_i}} + p_{b_i}$.
    Since $p_{b_i} / q_{b_i} \le \beta$ and $\abs{W_{t_i}} / t_i \le \beta$, we have
    \[\frac{\abs{W_{t_{i+1}}}}{t_{i+1}} \le \frac{\abs{W_{t_i}} + p_{b_i}}{t_{i} + q_{b_i}} \le \beta,\]
    and the induction is complete.

    However, within stage $i$, the first $p_{b_i}$ strings output by the algorithm may be incorrect.
    Thus, the maximum fraction of errors during stage $i$ is bounded by
    \[\frac{\abs{W_{t_i}} + p_{b_i}}{t_{i} + p_{b_i}} \le
    \frac{\beta t_i + p_{b_i}}{t_{i} + p_{b_i}} \le \beta + \frac{p_{b_i}}{t_i}.\]
    To bound the term $p_{b_i} / t_i$, notice that $b_i \le 2t_i + 2$, as
    at least one of the first $2t_i + 2$ boxes must contain no adversary or algorithm outputs
    at time $t_i + 1$.
    Now recall that $\lim_{i \to \infty} p_{2i+2} / i = 0$,
    so
    \[\limsup_{t \to \infty} \frac{\abs{W_t}}{t} \le \limsup_{i \to \infty} \beta + \frac{p_{2i+2}}{t_i} = \beta,\]
    proving that the algorithm $\beta$-generates in the limit.
\end{proof}

We now prove the other direction of~\cref{thm:box}.
\begin{lemma}
    \label{lem:box_nexist}
    There does not exist an algorithm that $\beta'$-generates in the limit without repetitions for $C_{\beta}$
    for any $\beta' < \beta$.
\end{lemma}

\begin{proof}
    Let us reserve a set of time stamps $R = \{2^i \mid i \in \mathbb{N}\}$.
    For each box $S_i$, the adversary essentially wants to ensure that the first $p_i$
    outputs from the algorithm in that box are wrong.
    To do this, for each of the first $p_i$ times the algorithm ever outputs a string in $S_i$
    (which do not have to be consecutive time slots), the adversary ``declares'' that the string output by the algorithm
    is wrong and adds the string to a forbidden set $F_i$ that the adversary will never output from.
    We now describe the adversary's outputs.
    At some time $t$, let $S_i$ be the box that contains the algorithm's previous output $z_{t-1}$.
    If $t \notin R$ and the adversary has output fewer than $2q_i - p_i$ strings from $S_i$,
    the adversary sets $x_t$ to be any string from $S_i \setminus F_i$
    that the adversary has not yet output.
    Otherwise, if $t \in R$ or the algorithm has already output $p_i$ strings from $S_i$,
    the adversary finds the first box $S_j$ where the adversary has output fewer than $2q_j - p_j$ strings,
    and sets $x_t$ to be any string from $S_j \setminus F_j$ that it has not yet output.
    Finally, if this was among the first $p_i$ times that the algorithm output a string in $S_i$, the adversary adds
    $z_{t-1}$ to $F_i$.

    To see that any algorithm generates with at least $\beta$ fraction of errors,
    at some time $t$, partition the set of time-steps $\{1, 2, \dots, t\}$ into $T_1 \cup \dots$ so that each $T_i$
    contains exactly the time-stamps when the algorithms output a string in the box $S_i$ among the first $t$
    time stamps.

    Then the first $\min(p_i, |T_i|)$ time-steps in $T_i$ are always wrong.
    Note that since the adversary always tries to mirror the algorithm unless $t \in R$,
    the adversary always reveals at least $|T_i| - |R \cap T_i|$ strings in
    $S_i$ before the algorithm in the first $t$ timestamps.
    Thus $|T_i| - |R \cap T_i| + |T_i| \leq |S_i|$, implying $|T_i| \leq |S_i|/2 + |R \cap T_i|/2$.

    Therefore the hallucination rate is at least
    \[\limsup_{t \to \infty} \frac{\sum_{i} \min(p_i, |T_i|)} {\sum \min(|T_i|, |S_i| /2 + |R \cap T_i|/2)}
    \geq \limsup_{t \to \infty} \frac{\sum_{i} \min(p_i, |T_i|)} {\sum \min(|T_i|, |S_i| /2) + |R \cap T_i|/2}  \geq
    \beta \]
    because $\sum |R \cap T_i| = \log t$, and $\min(p_i, |T_i|) \geq \beta \min(|T_i|, |S_i| /2)$.
\end{proof}

Combining~\cref{lem:box_algo,lem:box_nexist} gives the proof of~\cref{thm:box}.

\subsection{Hierarchy of Generation with Breadth}

We next establish a strict hierarchy of generation at every rate of hallucination
and breadth of generation.

\densityhierarchywithout*

To prove \cref{thm:density_hierarchy_without}, we first introduce some gadget collections that,
for any rate of hallucination, allow generation at the extreme ends of the density spectrum --- with either
lower density $0$ or $1/2$.
Then, we interpolate between these gadget collections to prove \cref{thm:density_hierarchy_without}.
In the remaining lemmas in this section, all notions of generation
refer to the setting of generation without algorithm repetitions.

We introduce the first gadget, which is a collection that cannot be generated with upper density any higher than $0$. 
\begin{lemma}
    \label{lem:box_zero_density}
    For any $\beta \in (0, 1]$, there exists a collection $\SC_{\beta}$ that is $\beta$-generatable
    in the limit with lower density $0$, but is not $\beta'$-generatable in the limit for any $\beta' < \beta$
    or generatable with upper density $\alpha$ for any $\alpha > 0$.
\end{lemma}
\begin{proof}
    The collection $\SC_{\beta}$ is identical to the collection $\SC_{\beta}$ used in the proof of \cref{thm:box}.
    However, we will choose an ordering of the universe $U$
    which will give the desired lower density of generation.
    Let $H$ and $Y$ be two infinite sets which partition $\NN$ such that $H$ has upper and lower density $1$, and $Y$
    has upper and lower density $0$.
    Recall for each language $L \in \SC_{\beta}$ and each box $S_i$ of size $2q_i$, there are exactly $p_i$ strings from
    $S_i$ that are not contained in $L$.

    For each box $S_i$, we will arbitrarily map exactly $2 p_i$ strings to the density $1$ set $H$,
    and map the remaining strings of $S_i$ to the density $0$ set $Y$.

    The collection $\SC_{\beta}$ is identical to the collection $\SC_{\beta}$ used in the proof of \cref{thm:box},
    so there exists an algorithm that $\beta$-generates in the limit (with lower density $0$),
    and there does not exist an algorithm that $\beta'$-generates in the limit for any $\beta' < \beta$.

    We now show that an algorithm cannot generate at an upper density higher than $0$. We define an adversary strategy.
    Let \( F = \{F_1, F_2, \dots\} \) denote the set of timestamps where \(F_i\)
    is the first time the algorithm has touched \(2^i\) boxes.
    At some time $t$, let $S_i$ be the box that contains the algorithm's previous output $z_{t-1}$.
    At any time $t \notin F$, the adversary outputs an unrevealed string in $S_i$ mapped to $H$,
    if such a string exists, and otherwise outputs an arbitrary unrevealed string.
    At the times $t \in F$, the adversary fills in a string from the first box $S_i$
    that is missing a string from the true language $K$.
    At all times, if it is among the first $p_i$ times that the algorithm
    outputs a string in $S_i$ that is mapped to $H$,
    the adversary declares that string to be wrong and will never output the string in the future.
    Therefore for each $S_i$, since the adversary is mirroring the algorithm at all times other than in $F$,
    the algorithm generates at most $M_i$ strings in $S_i$ that are mapped to $H$,
    where $M_i$ is the number of times in which the algorithm output a string $z_t \in S_i$ and $t+1 \in F$.
    The true language in $S_i$ should contain at least $p_i$ strings mapped to $H$ except those time stamps in $F$.
    Notice among the first $i$ boxes, there are at most $\log i$ time stamps in $F$,
    implying that $\sum_{j<i} M_j = \log i$.
    Therefore the generation upper density is at most
    \[\limsup_{i \to \infty} \frac{\log (i+1)}{\sum_{j \leq i} p_j - \log i} = 0. \qedhere\]
\end{proof}

The next two gadgets are $\beta$-generatable with lower density $1/2$.
We will have two different collections for the regimes of $\beta \in (0, 1/2]$ and $\beta \in (1/2, 1]$.

\begin{lemma}
    \label{lem:box_half_density}
    For any $\beta \in (0, 1/2]$, there exists a collection $\SC_{\beta}$
    that is $\beta$-generatable in the limit with lower density $1/2$,
    but is not $\beta'$-generatable in the limit for any $\beta' < \beta$ regardless of the breadth of generation.
\end{lemma}
\begin{proof}
    The collection $\SC_{\beta}$ will be similar to the collection from the proof of \cref{thm:box}.
    Now, each box $S_i$ will contain a set of $q_i$ (not $2q_i$) bins $B_{ij}$, each of which is simply a set of $2$
    strings.
    For any bin $B_{ij}$, each language $L \in \SC_{\beta}$
    will either contain both strings from the bin or neither of the two strings.
    Each language will contain the strings from exactly $q_i - p_i$ bins from each box $S_i$.
    Thus, each box $S_i$ now contains $2q_i$ strings, and each language contains $2(q_i - p_i)$ strings from $S_i$.
    The collection $\SC_{\beta}$ is the set of languages formed from all possible ways of taking exactly $q_i - p_i$
    bins from each box $S_i$.
    We order the strings in the universe in increasing order, first by the index of the boxes,
    and then by the index of the bins within each box.

    We first give an algorithm that $\beta$-generates in the limit with lower density $1/2$.
    At any time $t$, if $x_t$ is in some bin $B_{ij}$, and the other
    string of $B_{ij}$ is unrevealed, the algorithm mirrors the adversary by outputting that other string.
    Otherwise, the algorithm finds the first box $S_i$ in which there exists a bin $B_{ij}$
    such that the algorithm has not output any string from $B_{ij}$, and $B_{ij}$
    contains an unrevealed string.
    The algorithm then outputs the unrevealed string from $B_{ij}$.

    First, it is clear that the algorithm generates with lower density $1/2$.
    Consider any bin $\{y_1, y_2\}$ that is contained in the target language $K$.
    The adversary must eventually output both $y_1$ and $y_2$ in its enumeration.
    For the first string the adversary reveals, the other one should be immediately generated by the algorithm,
    or has already been output by the algorithm.
    Thus, for every bin $\{y_1, y_2\} \subseteq K$, the algorithm will occupy at least one of the strings,
    implying lower density at least $1/2$.

    We now show that the algorithm $\beta$-generates in the limit.
    Clearly any algorithm outputs produced by mirroring the adversary must be correct.
    For any time $t$, let $s(t)$ be the index of the smallest box
    containing a bin $B_{s(t),j}$ without any algorithm outputs.
    Note that the adversary also cannot have output any strings from $B_{s(t),j}$, because otherwise,
    the algorithm would have mirrored the adversary and also output a string from $B_{s(t),j}$.
    In particular, this implies that any algorithm outputs up to $t$ in later bins (location-wise)
    were produced by mirroring the adversary, and thus are correct.
    Thus, the only incorrect algorithm outputs must have come from $S_1$, \dots, $S_{s(t)}$.
    Since the algorithm outputs at most one string from each bin, there are at most $\sum_{i = 1}^{s(t)} p_i$
    incorrect time-steps.
    By definition of $s(t)$, every bin from $S_1$, \dots, $S_{s(t)-1}$ contains at least one algorithm output,
    implying that the algorithm has output at least $\sum_{i = 1}^{s(t)-1} q_i$ total strings.
    Thus the fraction of error is at most
    \[\limsup_{t \to \infty} \frac{\sum_{i = 1}^{s(t)} p_i}{\sum_{i = 1}^{s(t)-1} q_i} \le \beta,\]
    since $p_i/q_i$ converges to $\beta$ and $p_i/\sum_{j=1}^{i} q_j$ converges to $0$.

    Finally, we show that no algorithm can generate with hallucination rate less than $\beta$.
    The adversary follows a similar strategy as in \cref{lem:box_nexist}.
    We again reserve a set of time-steps $R = \{2^i \mid i \in \NN\}$.
    We designate each of the first $p_i$ distinct bins that the algorithm outputs from in some box $S_i$,
    as being forbidden (as long as the adversary has not already output from that bin),
    and the adversary will never output from those bins.
    At some time $t$, let the algorithm's previous output $z_{t-1}$ lie in some bin $B_{ij}$
    inside box $S_i$. 
    At a non-cleanup time $t \notin R$, if the adversary has not yet output the other string in the bin $B_{ij}$,
    and the bin is not forbidden, the adversary outputs that other string.
    Otherwise, if the bin is forbidden, or the adversary has already output the other string in $B_{ij}$ or $t \in R$,
    the adversary finds the first box $S_j$ such that the adversary has touched fewer than $q_j - p_j$
    strings from $S_j$ and such that $S_j$ contains some bin that is not forbidden
    and has not been fully enumerated by the adversary.
    The adversary then outputs the unenumerated string from the appropriate bin in $S_j$.
    The target language is $K = \{x_1, x_2, \dots\}$.
    Since $R$ is infinite, in each box $S_i$, the adversary eventually outputs both strings
    from exactly $q_i - p_i$ bins and no strings from the other $p_i$, implying that $K \in \SC_{\beta}$.

    The argument that any algorithm generates with at least $\beta$ fraction of errors,
    is similar to the argument in \cref{lem:box_nexist}.
    At some time $t$, partition the set of time-steps $\{1, 2, \dots, t\}$ into $T_1 \cup \dots$ so that each $T_i$
    contains exactly the time-stamps when the algorithm output a string in the box $S_i$ among the first $t$
    time stamps.
    Furthermore, let $M_i$ be the set of cleanup times during which the adversary output a string in $S_i$.
    Note that the first $p_i$ strings that the algorithm outputs in each box are incorrect,
    except possibly if the adversary has already output from that bin during a cleanup time.
    Thus the number of wrong strings from $S_i$ is at least $\min(p_i - \abs{M_i}, \abs{T_i})$.
    The adversary always reveals at least $|T_i| - |R \cap T_i|$ strings in
    $S_i$ before the algorithm in the first $t$ timestamps.
    Thus $|T_i| - |R \cap T_i| + |T_i| \leq |S_i|$, implying $|T_i| \leq |S_i|/2 + |R \cap T_i|/2$.

    Therefore the hallucination rate is at least
    \[\limsup_{t \to \infty} \frac{\sum_{i} \min(p_i - \abs{M_i}, |T_i|)}{\sum_i \min(|T_i|, |S_i| /2 + |R \cap T_i|/2)}
    \geq \limsup_{t \to \infty} \frac{\sum_{i} \min(p_i, |T_i|) - \sum_{i} \abs{M_i}}
    {\sum_i \min(|T_i|, |S_i| /2) + |R \cap T_i|/2} \geq \beta \]
    because $\sum_i |R \cap T_i| = \sum_{i} \abs{M_i} = \log t$, and
    $\min(p_i, |T_i|) \geq \beta \min(|T_i|, |S_i| /2)$.
\end{proof}

\begin{lemma}
    \label{lem:tree_half_density}
    For any $\beta \in (1/2, 1]$, there exists a collection $\SC_{\beta}$
    that is $\beta$-generatable in the limit with lower density $1/2$,
    but is not $\beta'$-generatable in the limit for any $\beta' < \beta$ regardless of the breadth of generation.
\end{lemma}
\begin{proof}
    The collection $\SC_{\beta}$ is exactly equal to the tree construction from \cref{thm:real},
    except we take $\{p_i\}$ and $\{q_i\}$ to be sequences
    such that $p_i/q_i$ converges to $2 \beta - 1$ instead of $\beta$.
    Note that this implies $(q_i - p_i)/q_i$ converges to $2 - 2\beta$.

    To define the ordering on the universe $U$, let $H$ and $Y$ be two infinite sets
    that partition $\NN$ such that $H$ has upper and lower density $1$ and
    $Y$ has upper and lower density $0$.

    Recall that at every node $v$ of depth $i$,
    there are two sets $B_v$ of size $q_i - p_i + 1$, and $A_v$ of size $2q_i$.
    The entire set $B_v$ is contained in every language that passes through $v$,
    while exactly $p_i - 1$ strings from $A_v$ will be contained in a language that passes through $v$.
    We map every string in $B_v$ to the high density set $H$
    and every string in $A_v$ to the low density set $Y$.
    The mapping is layer-order-preserving in $H$ and $Y$:
    for any $i<i'$, every image of a string in layer $i$ precedes every image of a string in layer $i'$ in $H$,
    and similarly for $Y$.

    The algorithm that $\beta$-generates in the limit without repetitions will be very similar
    to the algorithm from \cref{lem:real_algo} that $\beta$-generates in the limit (with repetitions).
    At some node $v$ with depth $i$, the algorithm can correctly generate at least $\abs{B_v} / 2 = (q_i - p_i)/2$ strings from $B_v$.
    Since the adversary can spend at most $q_i$ time steps generating from $D_v$, we have that the fraction of correct time steps at node $v$ is at least $(q_i - p_i)/2q_i$, which converges to $1 - \beta$ as desired.

    To see that the algorithm generates with lower density $1/2$, note that the algorithm ensures
    that at each node $v_i$, the algorithm generates at least half of the strings from $B_{v_i}$.
    Since only the strings in sets $B_{v_i}$ are mapped to the density $1$ set $H$,
    this implies that the algorithm generates with lower density $1/2$.

    Finally, to see that $\SC_{\beta}$ is not $\beta'$-generatable in the limit without repetitions
    for any $\beta' < \beta$, the adversary construction will follow identically to the proof in \cref{lem:real_nexist}.
    However, now at each layer $i$, since the algorithm is not allowed to repeat strings, the algorithm can only generate half of the nodes from the set $B_{v_i}$, implying that the algorithm generates correctly for $(q_i - p_i)/2$ time steps.
    The adversary will still spend $q_i$ total steps generating from $D_{v_i}$.
    As before, the adversary ensures that during this time, the algorithm does not generate any strings correctly from $A_{v_i}$, and does not correctly guess the next child $v_{i+1}$.
    Thus, the algorithm can generate correctly at most a
    $(q_i - p_i)/2q_i$ fraction of the time, which converges to $1 - \beta$ as desired.
\end{proof}

We now show how to interpolate between \cref{lem:box_zero_density}
and \cref{lem:box_half_density,lem:tree_half_density} to prove \cref{thm:density_hierarchy_without}.
For any two collections $\SC_1$ and $\SC_2$, we write
$\SC_1 \times \SC_2 = \{L_1 \cup L_2 \mid L_1 \in \SC_1, L_2 \in \SC_2\}$.

\begin{lemma}
    \label{lem:interpolate}
    For any $\beta \in (0, 1]$, let $\SC_1$ be a collection over $U_1$ and $\SC_2$ be a collection over $U_2$
    such that both $\SC_1$ and $\SC_2$ are $\beta$-generatable in the limit but not $\beta'$-generatable in the limit
    for any $\beta' < \beta$.
    Furthermore, assume that $\SC_1$ is not generatable with upper density greater than $0$,
    and $\SC_2$ is $\beta$-generatable with lower density $1/2$.
    Finally, assume that $U_1$ and $U_2$ are disjoint,
    and that for any $\alpha \in (0, 1/2]$ there exists a fixed ordering of $U_1 \cup U_2$
    such that for any languages $L_1 \in \SC_1$ and $L_2 \in \SC_2$,
    the language $L_2$ has $2 \alpha$ upper and lower density in $L_1 \cup L_2$.

    For any $\alpha \in (0, 1/2]$, by ordering the strings in $U_1 \cup U_2$ appropriately,
    the collection $\SC = \SC_1 \times \SC_2$ is $\beta$-generatable in the limit
    with lower density $\alpha$.
    Furthermore, for any $\beta' < \beta$, the collection $\SC$ is not $\beta'$-generatable in the limit.
    Also, the collection $\SC$ is not generatable with upper density $\alpha' > \alpha$
    regardless of the hallucination rate.
\end{lemma}
\begin{proof}
    Fix $\alpha \in (0, 1/2]$ and order the strings of $U_1 \cup U_2$
    so that for any languages $L_1 \in \SC_1$ and $L_2 \in \SC_2$,
    the language $L_2$ has $2 \alpha$ upper and lower density in $L_1 \cup L_2$.

    The algorithm that $\beta$-generates in the limit with lower density $\alpha$ for $\SC$ simply
    combines the generators for $\SC_1$ and $\SC_2$.
    Let $G_1$ be an algorithm that $\beta$-generates in the limit for $\SC_1$
    and $G_2$ be an algorithm that $\beta$-generates in the limit with lower density $1/2$ for $\SC_2$.
    For an enumeration $x_1$, \dots, $x_t$, let $x^{{(i)}}(t)$ be the enumeration
    restricted to the strings in $U_i$.
    At some time $t$, if $x_t \in U_1$, then the algorithm outputs $G_1(x^{(1)}(t))$.
    Otherwise, the algorithm outputs $G_2(x^{(2)}(t))$.
    Since the individual algorithms $\beta$-generate in the limit,
    clearly the combined algorithm must also $\beta$-generate in the limit.
    Let $K = K_1 \cup K_2$ be an arbitrary target language where $K_1 \in \SC_1$ and $K_2 \in \SC_2$.
    Furthermore, $G_2$ generates with lower density $1/2$,
    and $K_2$ has density $2 \alpha$ in $K$.
    This implies that the combined algorithm generates with density $2 \alpha/2 = \alpha$ in $\SC$ as desired.

    We now show the lower bounds.
    For any algorithm $G$, there exist adversaries $A_1$ and $A_2$
    which ensure that the projections of $G$ onto $U_1$ and $U_2$ do not $\beta'$-generate in the limit
    for any $\beta' > \beta$.
    We get an adversary $A$ for $\SC$ by combining $A_1$ and $A_2$.
    At some time $t$, if the algorithm $G$ outputs a string $z_t \in U_1$,
    run the adversary $A_1$ on the projected outputs of $G$ onto $U_1$.
    If $z_t \in U_2$, run $A_2$ on the projected outputs of $G$ onto $U_2$.
    Since $A_1$ and $A_2$ ensure that the projections of $G$ onto $U_1$ and $U_2$ do not $\beta'$-generate in the limit
    for any $\beta' > \beta$,
    the adversary $A$ ensures that $G$ does not $\beta'$-generate in the limit as well.

    Finally, since no algorithm can generate with better than $0$ upper density for $\SC_1$,
    there exists an adversary $A_1$ which ensures that the projected outputs of $G$ onto
    $U_1$ have upper density $0$.
    For $\SC_2$, no algorithm can generate with higher than $1/2$ upper density.
    Thus, in total, since $K_2$ has density $2 \alpha$ in $K$,
    no algorithm can generate with higher than $2 \alpha/2 = \alpha$ upper density as desired.
\end{proof}

Using \cref{lem:interpolate}, we now interpolate between the previous gadgets
to prove \cref{thm:density_hierarchy_without}.

\begin{proof}[Proof of \cref{thm:density_hierarchy_without}]
    Fix an arbitrary $\beta \in (0, 1]$.
    If $\beta \in (0, 1/2]$, we apply \cref{lem:interpolate}
    by taking $\SC_1$ and $U_1$ to be the collection and universe from \cref{lem:box_zero_density},
    and $\SC_2$ and $U_2$ to be the collection and universe from \cref{lem:box_half_density}.
    It remains to show that for any $\alpha \in (0, 1/2]$,
    there exists an ordering of $U_1 \cup U_2$ satisfying the conditions of \cref{lem:interpolate}.

    Fix $\alpha \in (0, 1/2]$ and let $\{r_i\}$ and $\{s_i\}$ be any sequences of positive integers
    such that $s_i / (r_i + s_i)$ converges to $2 \alpha$ from below and
    $(r_i + s_i)/(\sum_{j \le i} r_j + s_j)$ converges to $0$.
    Intuitively, we will arrange the strings from $U_1 \cup U_2$ into groups $M_i$,
    so that at each group $M_i$, and for any $L_1 \in \SC_1$ and $L_2 \in \SC_2$,
    we roughly have $\abs{L_1 \cap M_i} = r_i$ and $\abs{L_2 \cap M_i} = s_i$.
    Since $s_i / (r_i + s_i)$ converges to $2 \alpha$, this would imply the desired result.

    We will inductively construct the strings in each group $M_i$ by
    placing strings from $U_1$ and $U_2$ sequentially into the groups.
    For $U_1$, we will only insert into the groups strings from $U_1$
    that are mapped to the high density set $H$ within $U_1$.
    We will insert the strings that map to the density $0$ set $Y$ at the end of the process.
    Assume inductively that we are at some group $M_i$ and have appropriately placed strings into $M_{i'}$ for $i' < i$.
    Also assume that we are at some box $S_j \in \SC_1$ and have placed all strings from previous boxes that map to $H$
    into the previous groups.
    Note that for box $S_j$, a language $L_1 \in \SC_1$ contains
    $2q_j - p_j$ out of the $2q_j$ total strings from $S_j$.
    Thus, we place into $M_i$ the next $\floor{r_i 2q_j/(2q_j-p_j)}$ strings from $S_j$
    that map to $H$ (which have not already been placed into previous groups).
    If there are not enough new strings from $S_{j}$, we move on to the next box and fill the required strings from
    $S_{j+1}$.
    Note that in the limit, since a language $L_1 \in \SC_1$ contains $2q_j - p_j$
    out of the $2q_j$ total strings from a box $S_j$,
    and because the sequence $\{q_i\}$ and $\{r_i\}$ do not grow too quickly,
    we have that $\abs{L_1 \cap M_i}$ converges to $r_i$.

    We place strings from $U_2$ in a similar fashion.
    For each box $S_j \in \SC_2$, a language $L_2 \in \SC_2$
    contains $2q_j - 2p_j$ out of the $2q_j$ total strings from $S_j$.
    Thus, we place into $M_i$ the next $\floor{s_i 2q_j/(2q_j-2p_j)}$ strings from $S_j$,
    filling from subsequent boxes if there are not enough new strings from $S_j$.
    As before, in the limit, for any language $L_2 \in \SC_2$,
    we have that $\abs{L_2 \cap M_i}$ converges to $s_i$.
    Finally, we place the strings from $U_1$ that map to the density $0$ set $Y$.
    These strings will be inserted so that the projection of the order of $U_1 \cup U_2$ onto $U_1$ respects
    the initial ordering on $U_1$, and the set of strings that map to $Y$ forms a density $0$ set in $U_1 \cup U_2$.
    This will not affect the density of a language $L_2$ in $L_1 \cup L_2$.
    
    As previously argued, since $s_i/(r_i + s_i)$ converges to $2 \alpha$,
    this implies that $L_2$ will have lower density exactly $2 \alpha$ in $L_1 \cup L_2$.
    By applying \cref{lem:interpolate}, this proves the theorem for $\beta \in (0, 1/2]$.
    
    If $\beta \in (1/2, 1]$, we apply \cref{lem:interpolate}
    by taking $\SC_1$ and $U_1$ to be the collection and universe from \cref{lem:box_zero_density},
    and $\SC_2$ and $U_2$ to be the collection and universe from \cref{lem:tree_half_density}.
    Again fix $\alpha \in (0, 1/2]$ and let $\{r_i\}$ and $\{s_i\}$ be the same sequences from before.
    We will again arrange the strings from $U_1 \cup U_2$ into groups $M_i$,
    so that at each group $M_i$, and for any $L_1 \in \SC_1$ and $L_2 \in \SC_2$,
    we roughly have $\abs{L_1 \cap M_i} = r_i$ and $\abs{L_2 \cap M_i} = s_i$.
    The strings from $U_1$ will be placed into $M_i$ identically to before.

    For $U_2$, note that the collection $\SC_2$ is defined by a tree where at each layer $i$,
    a language $L_2 \in \SC_2$ contains exactly $q_i - p_i + 1$ strings from a set $B_{v}$
    and $p_i - 1$ strings from a set $A_v$ where $v$ is some node at depth $i$.
    Furthermore, the universe of $U_2$ is ordered so that the strings in $B_v$
    are drawn from a density $1$ set, while the strings in $A_v$ are drawn from a density $0$ set.
    Thus, we will sequentially move down the tree by layers, and place only strings from sets $B_{v}$
    into the groups.
    Then, at the end, we will order the strings of the sets $A_v$
    so that they form a zero density set in $U_1 \cup U_2$.

    Assume inductively that we are at some group $M_i$
    and have appropriately placed strings from $U_2$ into $M_{i'}$ for $i' < i$.
    Also assume that we are at some depth $j$ in the tree for $U_2$, and that for all nodes
    $v$ with depth less than $j$, we have placed $B_v$ into some group $M_{i'}$.
    Now, for \emph{each} node $v$ at depth $j$, we will place exactly $s_i$ strings from $B_v$
    (that have not already been placed into groups) into $M_{i}$,
    moving down to the next layer in the tree if there are insufficient strings at depth $j$.
    Note that in this fashion, since $\{q_i\}$ and $\{s_i\}$ do not grow too quickly,
    we have that $\abs{L_2 \cap M_i}$ converges to $s_i$.

    Finally, after placing all of the strings from $U_1$ and all strings from $U_2$
    that are drawn from some set $B_v$ into groups, we insert the strings from the sets $A_v$
    into the ordering of $U_1 \cup U_2$ so that they form a density $0$ set within the ordering
    and the projection of the ordering of $U_1 \cup U_2$ onto $U_2$ respects the initial ordering of $U_2$.
    This will not affect the density of a language $L_2$ in $L_1 \cup L_2$.
    Thus, as argued previously, since $s_i/(r_i + s_i)$ converges to $2 \alpha$,
    this implies that $L_2$ will have lower density exactly $2 \alpha$ in $L_1 \cup L_2$.
    By applying \cref{lem:interpolate}, this proves the theorem for $\beta \in (1/2, 1]$.
\end{proof}

    \section{Closing Remarks}\label{sec:closing}In this paper, we initiated the study of language generation in the limit with hallucination, and showed several surprising phenomena. We first showed that allowing hallucination expands the set of generatable collections beyond the KM model, even when the hallucinated time-steps have $0$-measure in the time horizon. We then proceeded to show a fine-grained and strict hierarchy on language collections by the allowed rate of hallucination. We further showed that the hierarchy extends to breadth of generation, namely that for every density level and rate of hallucination, there is a language collection that is generatable with those parameters but not with any stricter ones. (The impossibility holds in a rather strict sense: even if one parameter is made stricter by an arbitrarily small amount, and the other parameter is relaxed to an arbitrarily large extent, the language collections are no longer generatable.) Finally, we considered generation without repetition, which allows us to translate the time-steps with correct/incorrect output to the set of correct/incorrect strings generated by the algorithm. Once again, we showed that the rate of hallucination creates a surprisingly fine-grained hierarchy: for every rate of hallucination, there is a language collection that is generatable at that rate but not at a lower rate. Similar to the case of generation with repetition, the strict hierarchy can also be extended in generation without repetition to include breadth of generation.

Our work opens several interesting directions for future research. From a technical perspective, an interesting question is to further explore the tradeoff between the error rate and breadth of generation. In practice, they correspond to two important phenomena observed in even the most advanced language models: {\em hallucination} and {\em mode collapse}. For instance, when can an algorithm leverage a higher rate of hallucination to increase the breadth of generation? We showed that there are language collections where this is not possible, but is the contrary also possible, i.e., are there language collections where hallucination leads to better density? If both phenomena are possible, then it would be interesting to characterize the role of hallucination in terms of improving breadth more precisely. Another intriguing high-level question is to develop a complexity theory of language generation. So far, the literature has largely focused on computability of language generation in the limit, i.e., a generation algorithm is not parametrized or constrained by the computational resources that it consumes. We believe it would be extremely interesting to go beyond computability to develop a complexity theory for this important phenomenon.

\paragraph{AI Disclosure.} 

We used ChatGPT $5.4$ Pro to generate \Cref{fig:zero-tree,fig:real-tree}. The authors edited the figures and verified their correctness.



    {\small
    \bibliographystyle{alphaurl}
    \bibliography{ian}
    }

\end{document}